\documentclass[a4paper,11pt,reqno]{amsart}

\usepackage{amsmath}
\usepackage[T1]{fontenc}
\usepackage{amssymb}
\usepackage{amsthm}
\usepackage{color}
\usepackage[lmargin=2.5 cm,rmargin=2.5 cm,tmargin=3.5cm,bmargin=2.5cm,paper=a4paper]{geometry}

\newcommand{\Ab}{\mathbf A}
\newcommand{\Fb}{\mathbf F}
\newcommand{\Eb}{\mathbf E}

\newcommand{\Rb}{\mathbf R}
\newcommand{\R}{\mathbb R}
\newcommand{\Z}{\mathbb Z}
\newcommand{\C}{\mathbb C}
\newcommand{\bb}{\mathfrak b}
\newcommand{\tchi}{\widetilde{\chi}}

\DeclareMathOperator{\E0}{E_{\text{g.st}}}
\DeclareMathOperator{\curl}{curl}
\DeclareMathOperator{\Div}{div}

\DeclareMathOperator{\dist}{dist}
\DeclareMathOperator{\supp}{supp}

\DeclareMathOperator{\err}{err}
\DeclareMathOperator{\Real}{Re}

\newtheorem{thm}{Theorem}[section]

\newtheorem{lem}[thm]{Lemma}
\newtheorem{corol}[thm]{Corollary}

\newtheorem{lemma}[thm]{Lemma}

\theoremstyle{remark}
\newtheorem{rem}[thm]{Remark}

\numberwithin{equation}{section}

\title[3D Ginzburg-Landau  functional]{The ground state energy of the three dimensional Ginzburg-Landau functional\\
{\small
Part~II: Surface regime}}
\author{S. Fournais}
\address[S. Fournais]
{Department of Mathematical Sciences, University
  of Aarhus, Ny Munkegade, Building
  1530, DK-8000 \AA rhus C, Denmark}
  \email{fournais@imf.au.dk}

\author{A. Kachmar}
\address[A. Kachmar]{Lebanese University, Department of Mathematics, Hadath, Lebanon.}
\email{ayman.kashmar@liu.edu.lb}

\author{M. Persson}
\address[M. Persson]{Centre for Mathematical Sciences, Box 118, SE-22100, Lund, Sweden.}
\email{mickep@maths.lth.se}


\begin{document}

\begin{abstract}
We study the Ginzburg-Landau model of superconductivity in three
dimensions and for strong external magnetic fields. For magnetic
field strengths above the phenomenologically defined second
critical field it is known from Physics that superconductivity
should be essentially restricted to a region near the boundary. We
prove that the expected region does indeed carry
superconductivity. Furthermore, we give precise energy estimates
valid also in the regime around the second critical field which
display the transition from bulk superconductivity to surface
superconductivity.
\end{abstract}

\maketitle

\section{Introduction and main results}

The phenomenological Ginzburg-Landau theory of superconductivity
successfully describes the behavior of a superconductor subject to an external
magnetic field. Also, as pointed out in the celebrated work of Abrikosov,
this theory predicted the existence of type~II superconductors before they had
been empirically  realized, see \cite{dGe} for a review of this physical topic for which A. Abrikosov was awarded the Nobel Prize.

In the Ginzburg-Landau theory, the superconducting state of a sample is
described by a complex-valued wave function $\psi$ and a vector field
(magnetic potential) $\Ab$ such that the pair $(\psi,\Ab)$ is a critical point
of a specific energy (see~\eqref{eq-3D-GLf} below). The physical interpretation
of $\psi$ and $\Ab$ is explained by the microscopic
Bardeen-Cooper-Schrieffer~(BCS) theory as follows: $|\psi|^2$ is proportional
to the density of superconducting particles and $\curl\Ab$ measures the induced
magnetic filed inside the sample. The rigorous mathematical justification of
the connection between the Ginzburg-Landau and the BCS theory has only been
established recently in \cite{FHSS}.

The behavior of a type~II superconductor is distinguished by three critical
values the intensity of the applied magnetic field can have, that we denote by
$H_{C_1}$, $H_{C_2}$ and $H_{C_3}$. These critical fields may be described
in terms of the wave function $\psi$ as follows. Suppose $H$ is the intensity
of the external magnetic field applied to the sample. If $H<H_{C_1}$, the
material is in the superconducting phase, which corresponds to $|\psi|>0$
everywhere. If $H_{C_1}<H<H_{C_2}$, the magnetic field penetrates the sample in
quantized vortices (corresponding to zeros of  $\psi$). If $H_{C_2}<H<H_{C_3}$, superconductivity is confined to (part of)
the surface of the sample (corresponding to $|\psi|$ very small in the bulk).
Finally, if $H>H_{C_3}$, superconductivity is lost, which is reflected by
$\psi=0$ everywhere.

In the last two decades, much progress has been made in order to
establish the aforementioned behavior of Type~II superconductors
by studying minimizers of the Ginzburg-Landau energy. The monograph \cite{SS} and references therein
contains an analysis of vortices and the critical field $H_{C_1}$.
Concerning the analysis of the critical fields $H_{C_2}$ and
$H_{C_3}$ we mention \cite{FK, FH-b} (and references therein). As
one can see in \cite{FH-b, SS}, the Ginzburg-Landau model has a
rich mathematical structure whose analysis  requires a diversity
of methods, and many of them have been developed especially for the
study of the model.

While a detailed study of the Ginzburg-Landau model in a two
dimensional domain has been the subject of numerous papers, the
study of the model in a three dimensional domain is much less
developed. Among the important problems that are still open in 3D
is the calculation of the critical field $H_{C_1}$ for general
domains\,\footnote{In a recent paper \cite{Jetal}, compactness
results for the 3D functional valid in general domain are
obtained, which allow to obtain the leading order term of
$H_{C_1}$. Earlier results include a  candidate for the expression
of $H_{C_1}$ in the case of the ball \cite{ABM}, and an expression
of $H_{C_1}$ in `thin' shell domains in \cite{CS}.}. Also, when
comparing with the existing results for 2D domains, a precise
localization of the wave-function $\psi$ is absent when the
external applied magnetic field varies from $H_{C_1}$ up to
$H_{C_3}$. However, for both 2D and 3D domains, a sharp
characterization of the critical field $H_{C_3}$ is given in
\cite{FH3d}. In this paper, together with \cite{FK3D, K3D}, we
give a detailed description of the behavior of the wave-function
$\psi$ and its energy for external magnetic fields varying  in the
range above $H_{C_1}$ and up to $H_{C_3}$. The results in this
paper concern surface ($3$D) superconductivity and the transition
that happens close to $H_{C_2}$ from bulk to surface
superconductivity.

 Let $\Omega\subset\R^3$ be a bounded and open
set with smooth boundary which
 models a superconducting sample subject to an applied external magnetic field.
The energy of the sample is given by the Ginzburg-Landau
functional,
\begin{multline}\label{eq-3D-GLf}
\mathcal E^{\rm 3D}(\psi,\Ab)=\mathcal
E_{\kappa,H}^{\rm 3D}(\psi,\Ab)=
\int_\Omega\Bigl[
|(\nabla-i\kappa
H\Ab)\psi|^2-\kappa^2|\psi|^2+\frac{\kappa^2}{2}|\psi|^4\Bigr]\,dx\\
+\kappa^2H^2\int_{\R^3}|\curl\Ab-\beta|^2\,dx\,.
\end{multline}
Here $\kappa$ and $H$ are two positive parameters whose physical interpretation
are as follows, the number $\kappa$ is
a material parameter, and the number $H$ is the intensity of a constant
magnetic field externally applied to the sample. As explained earlier, $\psi$
is a wave function (order parameter) and $\Ab$ is the induced magnetic
potential. We take $\psi$ and $\Ab$ in convenient spaces as follows,
\[
\psi\in H^1(\Omega;\C)\,,\quad  \Ab\in
\dot{H}^1_{\Div,\Fb}(\R^3)\,,
\]
where $\dot H^1_{\Div,\Fb}(\R^3)$ is the space introduced in
\eqref{eq-3D-hs} below. Finally, $\beta$ is the profile of the external magnetic
field that we choose constant, $\beta=(0,0,1)$.

Let $\dot H^1(\R^3)$ be the homogeneous Sobolev space, i.e. the closure
of $C_c^\infty(\R^3)$ under the norm $u\mapsto\|u\|_{\dot
  H^1(\R^3)}:=\|\nabla u\|_{L^2(\R^3)}$. Let further
  $\Fb(x)=(-x_2/2,x_1/2,0)$. Clearly $\Div \Fb=0$.

We define the space,
\begin{equation}\label{eq-3D-hs}
\dot H^1_{\Div,\Fb}(\R^3)=\{\Ab~:~\Div \Ab=0\,,\quad\text{ and}\quad
\Ab-\Fb\in \dot H^1(\R^3)\}\,.
\end{equation}

Critical points $(\psi,\Ab)\in H^1(\Omega;\C)\times \dot
H^1_{\Div,\Fb}(\R^3)$ of $\mathcal E^{\rm 3D}$
satisfy the Ginzburg-Landau equations,
\begin{equation}\label{eq-3D-GLeq}
\left\{
\begin{array}{lll}
-(\nabla-i\kappa H\Ab)^2\psi=\kappa^2(1-|\psi|^2)\psi&\text{in}& \Omega,
\\
\curl^2\Ab=-\displaystyle\frac1{\kappa H}\Im(\overline{\psi}\,(\nabla-i\kappa
H\Ab)\psi)\mathbf 1_\Omega&\text{in}& \R^3,\\
N\cdot(\nabla-i\kappa H\Ab)\psi=0&\text{on}&\partial\Omega\,,
\end{array}\right.\end{equation}
where $\mathbf 1_\Omega$ is the characteristic function of the domain
$\Omega$, and $N$ is the interior unit normal vector  of $\partial\Omega$.

For a solution $(\psi,\Ab)$ of~\eqref{eq-3D-GLeq}, the function
$\psi$ describes the superconducting properties of the material
and $H\curl\Ab$ gives the induced magnetic field.

The important
class materials called Type~II superconductors corresponds mathematically to the limit
$\kappa\to\infty$, see \cite{FH-b, SS}.

We define the ground state energy,
\begin{equation}\label{eq-3D-gs}
\E0 (\kappa,H)=\inf_{(\psi,\Ab)\in H^1(\Omega;\C)\times \dot
H^1_{\Div,\Fb}(\R^3)}\mathcal E^{\rm 3D}(\psi,\Ab)\,.
\end{equation}

The leading order asymptotics of the ground state energy involves a function
$g$ constructed in  \cite{FK3D} (the definition will be recalled in
Section~\ref{sec:E2}) satisfying that $g:[0,\infty) \rightarrow [-1/2,0]$ is
continuous, increasing and there is a constant $E_2<0$ with
\begin{align}
g(b) &=0,\qquad\qquad \qquad \qquad \qquad\forall \,b\geq 1,\\
g(b) &=E_2 (1-b)^2[1+ o(1)],\qquad \text{as } b \nearrow 1.
\end{align}

From \cite{FK3D} we have the following general (bulk) result as long as
$H/\kappa$ is bounded from below
\begin{align}\label{eq:energyAsymp}
\big| \E0
(\kappa,H) - g(\tfrac{H}{\kappa} )|\Omega| \kappa^2\big| \leq C \kappa ^{3/2}.
\end{align}
However, when $H/\kappa\geq 1$, $g(H/\kappa)=0$ and~\eqref{eq:energyAsymp}
only gives a somewhat weak estimate of the energy. The value $H=\kappa$
corresponds to the phenomenologically described critical field $H_{C_2}$
mentioned previously. So one expects that around this value superconductivity
should become concentrated near the boundary of the domain $\Omega$ and this
should be reflected in the energy asymptotics.

In this paper we will complete the study of the energy asymptotics for high
magnetic field initiated in \cite{FK3D} by
\begin{itemize}
\item giving the leading order energy asymptotics for $H>\kappa$, which will
be a surface energy of order of magnitude $\kappa$,
\item studying the transition from `bulk' to `surface' dominated energy at
$H \approx \kappa$. We will see that this transition takes place at $H-\kappa$
of order $\sqrt{\kappa}$ in the sense that
\begin{itemize}
\item If $H<\kappa - f(\kappa)$, where
$f(\kappa)/\sqrt{\kappa} \rightarrow +\infty$ as $\kappa \rightarrow \infty$,
then the leading contribution to the energy comes from the bulk.
\item If $H>\kappa - f(\kappa)$, where
$\limsup_{\kappa \rightarrow \infty} f(\kappa)/\sqrt{\kappa} \leq 0$ then the
leading contribution to the energy comes from the surface.
\item If $H=\kappa - a \sqrt{\kappa}$ for some constant $a>0$, then the bulk
and surface contributions of the energy have the same order of magnitude.
\end{itemize}
\end{itemize}

We now state the main results of the paper,
Theorems~\ref{thm-main}~and~\ref{thm-op}. These require the introduction of
some notation.

\begin{itemize}
\item If $x$ is a point on the boundary of $\Omega$, then $\nu(x)$ denotes the
angle in $[0,\pi/2]$ between the vector $\beta=(0,0,1)$ and the tangent plane
to $\partial\Omega$ at the point $x$.
\item  If $\nu\in[0,\pi/2]$,  $\zeta(\nu)$ is the lowest eigenvalue of a
magnetic Schr\"odinger operator in the half-space, see Section~\ref{sec:op}.
The function $\zeta$ is a continuous and strictly increasing bijection from
$[0,\pi/2]$ to $[\Theta_0,1]$ where $\Theta_0 \approx 0.59$ is a universal
constant (the definition of $\Theta_0$ will be recalled in~\eqref{eq-th0} below).
\item If $\nu\in[0,\pi/2]$ and $\bb\leq 1$, the constant $E(\bb
,\nu)$ will be introduced in Section~\ref{sec-E(nu)}. $E(\bb,\nu)$
depends continuously on both $\bb$ and $\nu$, vanishes when
$\zeta(\nu)\geq \bb$, and $E(\bb,\nu)<0$ otherwise.
\end{itemize}

Below is a statement of the main result concerning the asymptotic behavior of
the ground state energy.

\begin{thm}\label{thm-main}
Let $\mu:\R_+\to\R$ be a function satisfying
$\displaystyle\lim_{\kappa\to\infty}\mu(\kappa)=0$.
Then there exists a positive constant $\kappa_0$, and
a function $\err:\R_+\to\R$ such that
$\displaystyle\lim_{\kappa\to\infty}\err(\kappa)=0$ and the following is true.
If $\kappa\geq\kappa_0$ and $H\geq \kappa-\mu(\kappa)\kappa$\,, then the ground
state energy in~\eqref{eq-3D-gs} satisfies,
\[
\E0(\kappa,H)
=\sqrt{\kappa H}\,\int_{\partial\Omega}E\left(\bb,\nu(x)\right)\,d\sigma(x)+
E_2|\Omega|\,\left[\kappa-H\right]_+^2
+\err(\kappa)\max\left(\kappa,[\kappa-H]_+^2\right)\,.
\]
Here $\bb=\min\big(\kappa/H,1\big)$, and $d\sigma(x)$ is the surface measure
on the boundary of $\Omega$.
\end{thm}

The next theorem concerns the behavior of order parameters. We use
the convention that a subset $D\subset\Omega$  is smooth if there
exists an open subset $\widetilde D\subset\R^3$ having a smooth
boundary and such  that $D=\widetilde D\cap \Omega$.

\begin{thm}\label{thm-op}
Let $\mu:\R_+\to\R$ be a function satisfying
$\displaystyle\lim_{\kappa\to\infty}\mu(\kappa)=0$ and let $D\subset \Omega$
be smooth. Then there exist a positive constant $\kappa_0$, and a function
$\err:\R_+\to\R$ such that $\displaystyle\lim_{\kappa\to\infty}\err(\kappa)=0$,
and if $\kappa\geq\kappa_0$ and $H\geq \kappa-\mu(\kappa) \kappa$\,, then the
following is true.
\begin{enumerate}
\item If $(\psi,\Ab)\in H^1(\Omega;\C)\times \dot H^1_{\Div,\Fb}(\R^3)$ is a
solution of~\eqref{eq-3D-GLeq}, then,
\begin{multline}\label{eq-3D-op'}
\frac{1}2\int_D|\psi|^4\,dx
\leq -\kappa^{-1}\sqrt{ \frac{H}{\kappa}}\,
\int_{\partial\Omega\cap \overline{D}}
E\left(\frac{\kappa}H,\nu(x)\right)\,d\sigma(x)
-E_2|D|\,\left[\frac{\kappa}{H}-1\right]_+^2\\
+\err(\kappa)\max\left(\frac1\kappa,\left[\frac{\kappa}{H}-1\right]_+^2\right)\,.
\end{multline}
\item If $(\psi,\Ab)\in H^1(\Omega;\C)\times \dot H^1_{\Div,\Fb}(\R^3)$ is a
minimizer of~\eqref{eq-3D-GLf}, then,
\begin{multline}\label{eq-3D-op-Thm}
\frac12\int_D|\psi|^4\,dx= -\kappa^{-1}\sqrt{ \frac{H}{\kappa}}\,
\int_{\partial\Omega \cap \overline{D}}
E\left(\frac{\kappa}H,\nu(x)\right)\,d\sigma(x)
-E_2|D|\,\left[\frac{\kappa}{H}-1\right]_+^2\\
+\err(\kappa)\max\left(\frac1\kappa,\left[\frac{\kappa}{H}-1\right]_+^2\right)\,.
\end{multline}
\end{enumerate}
\end{thm}

\begin{rem}\label{rem-hc3}[Triviality for very large magnetic fields]~\\
Many previous works have addressed the question of the third
critical field, see \cite{FH-b} and the references therein for a
detailed review of this topic. We define
\begin{align}
H_{C_3}(\kappa) = \inf \big\{ H>0, &\text{ for all } H'>H \text{ the GL
equations~\eqref{eq-3D-GLeq} }\nonumber\\
&\text{ have only trivial solutions}\}.
\end{align}
Here a solution is trivial if $\psi \equiv 0$.
It follows from \cite{FH-b} that
\[
H_{C_3}(\kappa) = \Theta_0^{-1} \kappa + o(\kappa),\qquad \text{
as } \kappa \rightarrow \infty.
\]
Here, as we mentioned earlier, $\Theta_0 \approx 0.59$ is a
universal constant whose definition will be recalled in
\eqref{eq-th0} below.  So we can conclude that the ratio
$H/\kappa$ when $\kappa \geq 1$ is always uniformly bounded from
above in the regime where non-trivial solutions of
\eqref{eq-3D-GLeq} exist.
\end{rem}

It was realized early in the Physics literature (see for example
\cite[Chapter 6.6]{dGe}) that when studying the $3$-dimensional
Ginzburg-Landau model between $H_{C_2}$ and $H_{C_3}$ there is a
geometric boundary phenomenon that does not appear in $2$D. When
$\bb=\kappa/H \in (\Theta_0,1)$ (and $\kappa$ sufficiently large)
superconductivity will in the $2$D case exist on the entire
boundary. But in the $3$D model only a certain part of the
boundary will carry superconductivity. This part can be described
using the angle $\nu(x)$ (between the magnetic field and the
tangent plane at the boundary point $x$) and the spectral function
$\zeta$. A local and linearized calculation suggests that
superconductivity should only be present near the boundary section
\begin{align}
\Gamma(\bb) :=\{ x \in \partial \Omega \,:\, \zeta(\nu(x)) < \bb \}.
\end{align}
In previous works, notably \cite{Pa}, the technique of Agmon estimates was used
to prove that the superconducting order parameter $\psi$ will indeed decay
rapidly away from $\Gamma(\bb)$. Our results give the opposite direction,
namely exhibit through Theorem~\ref{thm-op} that minimizers are indeed nonzero
(in an $L^4$-sense) exactly in the vicinity of $\Gamma(\bb)$. In order to
realize this, notice that $[\kappa/H - 1]_{+}=0$ by the condition on $\bb$ and
that $E(\bb,\nu) = 0$ if and only if $\bb \leq \zeta(\nu)$.

\section{The universal constant $E_2$}\label{sec:E2}

Given a constant $b\geq 0$ and an open set $\mathcal D\subset \R^2$, we define
the following Ginzburg-Landau energy,
\begin{equation}\label{eq-LF-2D}
G_{\mathcal D}(u)=\int_{\mathcal D}\Bigl(b|(\nabla-i\Ab_0)u|^2
-|u|^2+\frac1{2}|u|^4\Bigr)\,dx\,.
\end{equation}
Here $\Ab_0$ is the canonical magnetic potential,
\begin{equation}\label{eq-hc2-mpA0}
\Ab_0(x_1,x_2)=\frac12(-x_2,x_1)\,,\quad\forall~x=(x_1,x_2)\in
\R^2\,.\end{equation}

Given $R>0$, we denote by  $K_R=(-R/2,R/2)\times(-R/2,R/2)$ a square of side
length $R$. Let,
\begin{equation}\label{eq-m0(b,R)}
m_0(b,R)=\inf_{u\in H^1_0(K_R;\C)}
G_{K_R}(u)\,.
\end{equation}

It is proved in \cite{FK3D} that
\begin{itemize}
\item For all $b>0$, there exists a constant $g(b)$
such that,
\[
\lim_{R\to\infty}\frac{m_0(b,R)}{R^2}=g(b)\,.
\]
\item The function $\displaystyle\frac{g(b)}{(b-1)^2}$ has a limit as $b\to1_-$,
\begin{equation}\label{eq-E2}
E_2=\lim_{b\to1_-}\frac{g(b)}{(b-1)^2}\,,
\end{equation}
and $-\frac12\leq E_2<0$.
\end{itemize}

\section{Reduced Ginzburg-Landau energy}\label{sec:hsp}

\subsection{Harmonic oscillator on the half-axis}
We denote by $\R_+=\{t\in\R~:~t>0\}$. For each real number $\xi\in\R$, we
consider the harmonic oscillator
\[
H(\xi)=-\frac{d^2}{dt^2}+(t-\xi)^2\quad\text{in}\quad L^2(\R_+)\,,
\]
with Neumann boundary condition, $u'(0)=0$.  Let $\mu_1(\xi)$ denotes the first
eigenvalue of $H(\xi)$. We define the universal constant,
\begin{equation}\label{eq-th0}
\Theta_0=\inf_{\xi\in\R}\mu_1(\xi)\,.
\end{equation}
The constant $\Theta_0$ satisfies the following properties \cite{DH, HM}:
\begin{equation}\label{min-th0}
\frac12<\Theta_0<1, \quad
\text{and}\quad \Theta_0=\mu_1(\xi_0),
\quad\text{where}\quad \xi_0=\sqrt{\Theta_0}\,.
\end{equation}
We introduce the function $\varphi_0\in L^2(\R^2)$ as follows,
\begin{equation}\label{eq-1D-phi0}
\int_{\R_+}|\varphi_0(t)|^2\,dt=1\,,
\quad -\varphi_0''(t)+(t-\xi_0)^2\varphi_0(t)=\Theta_0\varphi_0(t)
\quad \text{in }\R_+\,,
\quad \varphi_0'(0)=0\,.
\end{equation}

\subsection{Magnetic Schr\"odinger operator in $\R^3_+$}\label{sec:op}

We denote by $\R_+^3=\{(x_1,x_2,x_3)\in\R^3~:~x_1>0\}$. Let $\nu\in[0,\pi/2]$.
 Consider the magnetic potential
\begin{equation}\label{eq-3D-Eb}
\Eb_\nu(x)=(0,0,\cos\nu\,x_1+\sin\nu \,x_2)\,,\quad x=(x_1,x_2,x_3)\in \R_+^3\,.
\end{equation}
Notice that $\curl\Eb_\nu=(\sin\nu,-\cos\nu,0)$ is constant. Geometrically,
$\nu$ measures the angle between  $\curl\Eb_\nu$ and the boundary
of $\R_+^3$. Consider the Schr\"odinger operator with constant magnetic field\,,
\begin{equation}\label{eq-3D-op-2}
\mathcal L(\nu)=-(\nabla-i\Eb_\nu)^2\quad{\rm in}\quad L^2(\R_+^3)\,,
\end{equation}
with domain
\[
D(\mathcal L(\nu))
=\{u\in L^2(\R_+^3)~:~(\nabla-i\Eb_\nu)^ju\in L^2(\R_+)\,,~j=1,2\,,~
\partial_{x_1}u=0\text{ on }\{0\}\times\R^2\}\,.
\]
We denote by $\zeta(\nu)$ the bottom of the spectrum of $\mathcal L(\nu)$:
\begin{equation}\label{eq-z(nu)}
\zeta(\nu)=\inf\sigma\big{(}\mathcal L(\nu)\big{)}\,.
\end{equation}
We collect below some properties concerning $\zeta(\nu)$
(see e.g. \cite[Lemmas~7.2.1 \&  7.2.2]{FH-b}).

\begin{lem}\label{lem-p-z(nu)}
Let $\Theta_0$ be the universal constant introduced in~\eqref{eq-th0} above.
\begin{enumerate}
\item The function $[0,\pi/2]\ni\nu\mapsto\zeta(\nu)$ is monotone increasing,
$\zeta(0)=\Theta_0$ and $\zeta(\pi/2)=1$.
\item For all $\nu\in[0,\pi/2)$, $\zeta(\nu)<1$.
\end{enumerate}
\end{lem}

We denote by $\R^2_+=(0,\infty)\times\R$. In connection with the analysis of
the operator  $\mathcal L(\nu)$, we introduce the two dimensional operator
\begin{equation}\label{eq-2D-op}
L(\nu)=-\partial_{x_1}^2-\partial_{x_2}^2+(\cos\nu\,x_1+\sin\nu\,x_2)^2
\quad\text{in}\quad L^2(\R^2_+)\,,
\end{equation}
whose domain is
\[
D(L(\nu))=\{u\in H^2(\R_+^2)~:~(\cos\nu\,x_1+\sin\nu\,x_2)^ju\in L^2(\R^2_+),~j=1,2,~
\partial_{x_1}u=0~on ~\{0\}\times\R\}\,.
\]
The link between the spectra of $\mathcal L(\nu)$ and $L(\nu)$ is given
below~\cite{FH-b}.

\begin{lem}\label{lem-sp-L}
Suppose $\nu\in(0,\pi/2)$. Then,
\begin{enumerate}
\item $\sigma(\mathcal L(\nu))=\sigma(L(\nu))$\,;
\item $\sigma_{\text{ess}}(L(\nu))=[1,\infty)$\,.
\end{enumerate}
\end{lem}

\begin{rem}\label{rem-2D-sp}
Suppose that $\nu\in(0,\pi/2)$. It results from Lemma~\ref{lem-sp-L} and the
properties of $\zeta(\nu)$
in Lemma~\ref{lem-p-z(nu)} that
\begin{enumerate}
\item $\zeta(\nu)$ is an eigenvalue of finite multiplicity of $L(\nu)$\,;
\item $\zeta(\nu)$ is the lowest eigenvalue of $L(\nu)$.
\end{enumerate}
Consequently, we can select a real-valued eigenfunction $\phi\in L^2(\R^2_+)$
such that:
\[
\phi>0\,,\quad
\int_{\R^2_+}|\phi|^2\,dx=1\,,\quad
\int_{\R^2_+}\left(|\nabla\phi|^2+|(\cos\nu\,x_1+\sin\nu\,x_2)\phi|^2\right)\,dx
=\zeta(\nu)\,.
\]
Moreover, using the technique of `Agmon estimates' (cf. \cite{HM}), it is proved
in \cite[Theorem~1.1]{BoDaPoRa} that the eigenfunction $\phi$ decays
exponentially at infinity as follows. If $\nu\in(0,\pi/2)$ and
$\alpha\in (0,\sqrt{1-\zeta(\nu)})$, then  there exists a positive constant
$C_{\nu,\alpha}$ such that
\begin{equation}\label{eq:phidecay}
\int_{\R^2_+} e^{\alpha\sqrt{x_1^2+x_2^2}} \bigg(|\partial_{x_1}\phi|^2
+|\partial_{x_2}\phi|^2
+|(\cos\nu x_1+\sin\nu x_2)\phi|^2+|\phi|^2\bigg)\,dx_1dx_2
\leq C_{\nu,\alpha}.
\end{equation}
\end{rem}

\subsection{Reduced Ginzburg-Landau energy}
Let $\nu\in[0,\pi/2]$ and $\Eb_\nu$ be the magnetic potential introduced
in~\eqref{eq-3D-Eb}. For each $\ell>0$, we introduce the domains,
\begin{equation}\label{eq-K-ell}
K_\ell=(-\ell,\ell)\times(-\ell,\ell)\,,\quad  U_\ell=(0,\infty)\times K_\ell\,,
\end{equation}
and the space,
\begin{equation}\label{eq-domain-ell}
\mathcal S_\ell=\{u\in L^2(U_\ell)~:~(\nabla-i\Eb_\nu)u\in L^2(U_\ell)
\,,~u=0\text{ on }(0,\infty)\times \partial K_\ell\}\,.
\end{equation}
Let $\bb\in(0,1]$ be a given constant. If $u\in\mathcal S_\ell$,
we define the Ginzburg-Landau functional,
\begin{equation}\label{eq-Rgl}
\mathcal G_{\bb,\nu;\ell}(u)=
\int_{U_\ell}\left(|(\nabla-i\Eb_\nu)u|^2
-\bb|u|^2+\frac{\bb}2|u|^4\right)\,dx\,.
\end{equation}
Associated with $\mathcal G_{\bb,\nu;\ell}$ is the ground state energy,
\begin{equation}\label{eq-Rgs}
d(\bb,\nu;\ell)=\inf_{u\in \mathcal S_\ell} \mathcal G_{\bb,\nu;\ell}(u)\,.
\end{equation}

\subsubsection{Preliminary properties of minimizers}

\begin{lem}\label{lem-p-d(l)}
For each $\nu\in[0,\pi/2]$, let $\zeta(\nu)$ be as defined in~\eqref{eq-z(nu)}.
Let $\ell>0$ be given.

If $\bb\leq\zeta(\nu)$, then
\[
d(\bb,\nu;\ell)=0\,,
\]
where $d(\bb,\nu;\ell)$ is the ground state energy introduced in~\eqref{eq-Rgs}.
\end{lem}
\begin{proof}
By using $u=0$ as a test function, it is clear that
$d(\bb,\nu;\ell)\leq 0$. On the other hand, the min-max
principle and the condition on $b$ show that
$d(\bb,\nu;\ell)\geq0$.
\end{proof}

\begin{rem}\label{rem-b<theta0}
It results from Lemma~\ref{lem-p-d(l)} and the properties of $\zeta(\nu)$ in
Lemma~\ref{lem-p-z(nu)} that $d(\bb,\nu;\ell)=0$ in each of the following cases:
\begin{enumerate}
\item $\bb\leq\Theta_0$ and $\nu\in[0,\pi/2]$\,;
\item $\bb\in[\Theta_0,1]$ and $\nu=\pi/2$\,.
\end{enumerate}
\end{rem}

\begin{thm}\label{thm-min-Rgl}
Let $\Theta_0$ be the constant introduced in~\eqref{eq-th0}. Suppose that
$\bb\in[\Theta_0,1]$ is a given constant.

For all $\nu\in[0,\pi/2]$ and $\ell\geq1$, the functional
$\mathcal G_{\bb,\nu;\ell}$ in~\eqref{eq-Rgl} admits a minimizer
$\varphi_{\bb,\nu;\ell}\in \mathcal S_\ell$ satisfying,
\begin{equation}\label{eq-min-Rgl}
\mathcal G_ {\bb,\nu;\ell}(\varphi_{\bb,\nu;\ell})=d(\bb,\nu;\ell)
\,,\quad \|\varphi_{\bb,\nu;\ell}\|_{L^\infty(U_\ell)}\leq1\,.
\end{equation}
Furthermore, there exists a universal constant $C>0$ such that, if
$\nu\in[0,\pi/2]$ and $\ell>0$, the minimizer
$\varphi_{\bb,\nu;\ell}$ satisfies,
\begin{equation}\label{eq-min-decay}
\int_{U_\ell\cap\{x_1\geq 4\}}\frac{x_1}{(\ln x_1)^2}
\left(|(\nabla-i\Eb_\nu)\varphi_{\bb,\nu;\ell}|^2
+|\varphi_{\bb,\nu;\ell}|^2+x_1^2|\varphi_{\bb,\nu;\ell}|^4\right)
\,dx\leq C\ell^2\,.
\end{equation}
\end{thm}
\begin{proof}
Let $m>0$ and consider the energy functional,
\begin{equation}\label{eq-ApGL}
\mathcal G_{\ell,m}(u)=\int_{U_{\ell,m}}
\left(|(\nabla-i\Eb_\nu)u|^2-\bb|u|^2+\frac{\bb}2|u|^4\right)\,dx\,,
\end{equation}
where $\Eb_\nu$ is the magnetic potential  in~\eqref{eq-3D-Eb} and,
\[
U_{\ell,m}=(0,m)\times K_\ell\,,\quad K_\ell=(-\ell,\ell)\times(-\ell,\ell)\,.
\]
We define the ground state energy,
\begin{equation}\label{eq-t(l,m)}
t(\ell,m)=\inf_{u\in\mathcal S_{\ell,m}}\mathcal G_{\ell,m}(u)\,,\end{equation}
where
\begin{equation}\label{eq-S(l,m)}
\mathcal  S_{\ell,m}=\{u\in  H^1(U_{\ell,m})~:~u=0\text{ on }(0,m)
\times\partial K_\ell \text{ and }\{m\}\times K_\ell\}\,.
\end{equation}
The proof of Theorem~\ref{thm-min-Rgl} consists of showing that
$t(\ell,m)\to d(\bb,\nu;\ell)$ and that a minimizer $\varphi_{\ell,m}$ of
$\mathcal G_{\ell,m}$ converges to a minimizer of $\mathcal G_{\bb,\nu;\ell}$
as $m\to\infty$. That will be done in several steps.

\medskip
{\it Step~1.}
In this step, we prove that $t(\ell,m)\to d(\bb,\nu;\ell)$ as $m\to\infty$,
where $d(\bb,\nu;\ell)$ is the ground state energy introduced in~\eqref{eq-Rgs}.
Recall the space $\mathcal S_\ell$ introduced in~\eqref{eq-domain-ell}. Since
every $u\in \mathcal S_{\ell,m}$ can be extended by $0$ to a function
$\widetilde u\in\mathcal S_\ell$, we get that, $t(\ell,m)\geq d(\bb,\nu;\ell)$.
So, we need only prove that,
\begin{equation}\label{eq-t-to-d}
\limsup_{m\to\infty} t(\ell,m)\leq d(\bb,\nu;\ell)\,.
\end{equation}
Let $(\varphi_n)\subset\mathcal S_\ell$ be a minimizing sequence of
$\mathcal G_{\bb,\nu;\ell}$, i.e.
\[
d(\bb,\nu;\ell)=\lim_{n\to\infty}\mathcal G_{\bb,\nu;\ell}(\varphi_n)\,,
\]
where $\mathcal G_{\bb,\nu;\ell}(\varphi_n)$ is the functional introduced
in~\eqref{eq-Rgl}.

Let $\eta\in C_c^\infty(\R)$ be a cut-off function such that,
\[
0\leq\eta\leq1\text{ in }\R\,,\quad \eta=1\text{ in }[-1/2,1/2]
\,,\quad\text{supp}\,\eta\subset[-1,1]\,.
\]
Let $\eta_m(x_1)=\eta(x_1/m)$, then
$\eta_m(x_1)\varphi_n(x)\in\mathcal S_{\ell,m}$ and consequently,
\begin{equation}\label{eq-t<d}
t(\ell,m)\leq \mathcal G_{\bb,\nu;\ell}(\eta_m\varphi_n)\,.
\end{equation}
Using the inequality,
\[
|(\nabla-i\Eb_\nu)(\eta_m\varphi_n)|^2
\leq (1+\varepsilon)|\eta_m(\nabla-i\Eb_\nu)\varphi_n|^2
+2\varepsilon^{-1}|\nabla\eta_m|^2|\varphi_n|^2,
\]
valid for all $\varepsilon\in(0,1)$, together with  the fact that
$0\leq\eta_m\leq1$, we get the following estimate,
\begin{equation}\label{eq-t<d'}
t(\ell,m)\leq (1+\varepsilon)\mathcal G_{\bb,\nu;\ell}(\varphi_n)
+\frac{2\varepsilon^{-1}}{m^2}\|\eta'\|^2_{L^\infty(\R)}
\int_{U_\ell}|\varphi_n|^2\,dx+\bb\int_{U_\ell}\left(1-\eta_m^2
+\varepsilon\right)|\varphi_n|^2\,dx\,.
\end{equation}
Taking $\displaystyle\limsup_{m\to\infty}$ on both sides of~\eqref{eq-t<d'},
we get (by using in particular Lebesgue's dominated convergence theorem),
\[
\limsup_{m\to\infty}t(\ell,m)
\leq (1+\varepsilon)\mathcal G_{\bb,\nu;\ell}(\varphi_n)
+\bb\varepsilon\int_{U_\ell}|\varphi_n|^2\,dx\,.
\]
Taking successively $\varepsilon\to0_+$ then $n\to\infty$, we get the estimate
in~\eqref{eq-t-to-d}.

\medskip
{\it Step~2.} Since $U_{\ell,m}$ is a bounded domain, it is obvious that
$\mathcal G_{\ell,m}$ has a minimizer $\varphi_{\ell,m}$, see
e.g.~\cite[Chapter~11]{FH-b}. The function $\varphi_{\ell,m}$ satisfies the
following equation,
\begin{equation}\label{eq-st2-eq}
-(\nabla-i\Eb_\nu)^2\varphi_{\ell,m}
=\bb(1-|\varphi_{\ell,m}|^2)\varphi_{\ell,m}\,,\quad\text{ in }U_{\ell,m}\,.
\end{equation}
A simple application of the maximum principle yields
$|\varphi_{\ell,m}|\leq 1$ everywhere.
We will prove the following estimate,
\begin{equation}\label{eq-step2}
\int_{U_{\ell,m}\cap\{x_1\geq 4\}}\frac{x_1}{(\ln x_1)^2}
\left(|(\nabla-i\Eb_\nu)\varphi_{\ell,m}|^2
+|\varphi_{\ell,m}|^2+x_1^2|\varphi_{\ell,m}|^4\right)\,dx\leq C\ell^2\,,
\end{equation}
valid for all $\ell>0$, where  $C$ is a universal constant.

We obtain the estimate~\eqref{eq-step2} by following a construction similar
to \cite{P}. Consider a function $\chi\in C^\infty(\R)$ such that
$\supp\chi\subset(0,\infty)$.
Using the equation in~\eqref{eq-st2-eq} and an integration by parts, we get,
\begin{equation}\label{eq-step2-loc}
\int_{U_{\ell,m}}\left(|(\nabla-i\Eb_\nu)\chi\varphi_{\ell,m}|^2
-\bb|\chi\varphi_{\ell,m}|^2+\bb\chi^2|\varphi_{\ell,m}|^4\right)\,dx
=\int_{U_{\ell,m}}|\chi'|^2|\varphi_{\ell,m}|^2\,dx\,.
\end{equation}
Since $\chi\varphi_{\ell,m}\in H^1(\R^3)$ and the first eigenvalue of
the Schr\"odinger operator with constant unit magnetic field in $L^2(\R^3)$ is
equal to $1$, we get the lower bound,
\[
\int_{U_{\ell,m} }|(\nabla-i\Eb_\nu)\chi\varphi_{\ell,m}|^2\,dx
\geq \int_{U_{\ell,m}}|\chi\varphi_{\ell,m}|^2\,dx\,.
\]
Inserting this lower bound into~\eqref{eq-step2-loc} and remembering that
$\bb\in[\Theta_0,1]$, we deduce the following estimate,
\begin{equation}\label{eq-step2-L4}
\int_{U_{\ell,m}}\chi^2|\varphi_{\ell,m}|^4\,dx
\leq \frac{1}{\bb}\int_{U_{\ell,m}}|\chi'|^2|\varphi_{\ell,m}|^2\,dx\,.
\end{equation}
We select the function $\chi$ so that,
\begin{equation}\label{eq-zeta-1}
\chi(x_1)=0\text{~if~}x_1\leq0\,,\quad\chi(x_1)
=\frac{x_1^{3/2}}{\ln x_1}\text{~if~}x_1\geq 4\,.
\end{equation}
The function $\chi$ consequently satisfies
\begin{equation}\label{eq-zeta'}
0<\chi'(x_1)<\displaystyle\frac{3\sqrt{x_1}}{2\ln x_1}\quad
\text{for all } x_1\geq 4\,.
\end{equation}
Using the properties in~\eqref{eq-zeta-1} and~\eqref{eq-zeta'} of the
function $\chi$, we infer from~\eqref{eq-step2-L4},
\begin{align}\label{eq-step2-L4'}
\int_{U_{\ell,m}\cap\{x_1\geq4\}}
\frac{x_1^3}{(\ln x_1)^2}|\varphi_{\ell,m}|^4\,dx \nonumber\\
&\hskip-2cm\leq\frac1\bb\left[\frac94\int_{U_{\ell,m}\cap\{x_1\geq4\}}
\frac{x_1}{(\ln x_1)^2}|\varphi_{\ell,m}|^2\,dx+\int_{U_{\ell,m}\cap\{x_1\leq4\}}
|\chi'|^2|\varphi_{\ell,m}|^2\,dx\right]\nonumber\\
&\hskip-2cm\leq C\left(\int_{U_{\ell,m}\cap\{x_1\geq4\}}
\frac{x_1^3}{(\ln x_1)^2}|\varphi_{\ell,m}|^4\,dx\right)^{1/2}\ell+C\ell^2\,,
\end{align}
where $C$ is a universal constant.  We explain how we get the estimate
in~\eqref{eq-step2-L4'}. Actually, using that $|\varphi_{\ell,m}|\leq 1$, we
get,
\[
\int_{U_{\ell,m}\cap\{x_1\leq4\}}
|\chi'|^2|\varphi_{\ell,m}|^2\,dx\leq \|\chi'\|_{L^\infty([0,4])}(16\ell^2)\,.
\]
Also, using a Cauchy-Schwarz inequality, we have,
\begin{multline*}
\int_{U_{\ell,m}\cap\{x_1\geq4\}}\frac{x_1}{(\ln x_1)^2}|\varphi_{\ell,m}|^2\,dx\\
\leq
\left(\int_{U_{\ell,m}\cap\{x_1\geq4\}}\frac{1}{x_1(\ln x_1)^2}\,dx\right)^{1/2}
\left(\int_{U_{\ell,m}\cap\{x_1\geq4\}}
\frac{x_1^3}{(\ln x_1)^2}|\varphi_{\ell,m}|^4\,dx\right)^{1/2}\,,
\end{multline*}
thereby yielding~\eqref{eq-step2-L4'}.

As a consequence of~\eqref{eq-step2-L4'}, we deduce  the following estimate,
\begin{equation}\label{eq-step2-L4''}\
\forall~\ell\geq 1\,,\quad\int_{U_{\ell,m}\cap\{x_1\geq4\}}
\frac{x_1}{(\ln x_1)^2}
\left(|\varphi_{\ell,m}|^2+x_1^2|\varphi_{\ell,m}|^4\right)\,dx
\leq C\ell^2\,,\end{equation}
for a possibly new universal constant $C$. Thus, to finish the proof
of~\eqref{eq-step2}, we only need to prove that
\begin{equation}\label{eq-step2-der}
\int_{U_{\ell,m}\cap\{x_1\geq4\}}
\frac{x_1}{(\ln x_1)^2}|(\nabla-i\Eb_\nu)\varphi_{\ell,m}|^2\,dx\leq C\ell^2\,.
\end{equation}
To get~\eqref{eq-step2-der}, we select $\chi$ so that
$\chi(x_1)=\sqrt{x_1}/\ln x_1$ if $x_1\geq 4$, and $\chi(x_1)=0$ if $x_1\leq 1$.
Using the estimate,
\[
\int_{U_{\ell,m}\cap\{x_1\geq4\}}\chi^2|
(\nabla-i\Eb_\nu)\varphi_{\ell,m}|^2\,dx
\leq
2\int_{U_{\ell,m}}\left(|(\nabla-i\Eb_\nu)\chi\varphi_{\ell,m}|^2
+|\chi'|^2|\varphi_{\ell,m}|^2\right)\,dx\,,
\]
together with the estimates in~\eqref{eq-step2-loc},~\eqref{eq-step2-L4''} and
the properties of $\chi$, we get easily the estimate in~\eqref{eq-step2-der}.

\medskip{\it Step~3.} In this step, we prove that there exists a minimizer
$\varphi_\ell$ of $\mathcal G_{\bb,\nu;\ell}$ satisfying the estimate
in~\eqref{eq-min-decay}.

Starting from the bound $|\varphi_{\ell,m}|\leq 1$ and the
equation~\eqref{eq-st2-eq}, we get by a compactness and a diagonal sequence
argument that there exists a function $\varphi_\ell\in C^1(U_\ell)$ such that,
as $m\to\infty$,
\[
\varphi_{\ell,m}\to\varphi_\ell\quad\text{ in }\quad C^1(K)\,,
\]
for any compact set $K\subset \overline{U_\ell}$. Furthermore, $\varphi_\ell$
satisfies the equation,
\begin{equation}\label{eq-st3-eq}
-(\nabla-i\Eb_\nu)^2\varphi_{\ell}=\bb(1-|\varphi_{\ell}|^2)\varphi_{\ell}\,,
\quad\text{in }U_{\ell}\,,
\end{equation}
together with the boundary conditions,
\[
\frac{\partial\varphi_\ell}{\partial x_1}=0\text{ on }\{0\}
\times K_\ell\,,\quad \varphi_\ell=0\text{ on }(0,\infty)\times\partial K_\ell \,.
\]
Using the estimate in~\eqref{eq-step2} together with the bound
$|\varphi_{\ell,m}|\leq 1$, we get,
\[
\int_{U_{\ell,m}}\left(|(\nabla-i\Eb_\nu)\varphi_{\ell,m}|^2
+|\varphi_{\ell,m}|^2\right)\,dx\leq C\ell^2\,.
\]
Passing to a subsequence, we conclude that, as
$m\to\infty$, $\varphi_{\ell,m}\rightharpoonup\varphi_\ell$ and
$(\nabla-i\Eb_\nu)\varphi_{\ell,m}\rightharpoonup (\nabla-i\Eb_\nu)\varphi_\ell$
weakly  in  $L^2(U_\ell)$. Consequently, we get that $\varphi_\ell$ is in the
space $\mathcal S_\ell$,
\[
\int_{U_{\ell,m}}|\varphi_{\ell,m}|^4\,dx
\to \int_{U_\ell}|\varphi_\ell|^4\,dx\quad\text{ as }m\to\infty\,,
\]
and $\varphi_\ell$ satisfies the estimate in~\eqref{eq-min-decay}.

Using the equation~\eqref{eq-st2-eq} and an integration by parts, it is easy
to check that,
\[
t(\ell,m)=-\frac{\bb}2\int_{U_{\ell,m}}|\varphi_{\ell,m}|^4\,dx\,.
\]
Letting $m\to\infty$ and using Step~1 above, we get,
\[
d(\bb,\nu;\ell)=-\frac{\bb}2 \int_{U_\ell}|\varphi_\ell|^4\,dx\,.
\]
But, thanks to~\eqref{eq-st3-eq}, an integration by parts yields,
\[
\mathcal G_{\bb,\nu;\ell}(\varphi_\ell)
=-\frac{\bb}2 \int_{U_\ell}|\varphi_\ell|^4\,dx\,,
\]
thereby obtaining that $\varphi_\ell$ is a minimizer of
$\mathcal G_{\bb,\nu;\ell}$.
\end{proof}

\subsubsection{A rough estimate of the ground state energy}

\begin{lem}\label{lem-Rgs-B}
There exists a universal constant $C_1>0$ such that, for all
$\bb\in[\Theta_0,1]$, $\nu\in[0,\pi/2]$ and $\ell>0$, we have,
\[
d(\bb,\nu;\ell)\geq C_1\min(\zeta(\nu)-\bb,0)\ell^2\,.
\]
Furthermore, if
$\nu\in[0,\pi/2)$ satisfies $\zeta(\nu)<\bb$, there exist positive constants
$C_2$ and $\ell_0$ such that, if $\ell\geq\ell_0$, then,
\[
d(\bb,\nu;\ell)\leq -C_2\ell^2\]
\end{lem}

\begin{rem}\label{lem-Rgs-sup}
It results from Lemma~\ref{lem-Rgs-B} that
\[
-\infty<\inf_{\ell\geq\ell_0}\frac{d(\bb,\nu;\ell)}{\ell^2}
\leq \sup_{\ell\geq\ell_0}\frac{d(\bb,\nu;\ell)}{\ell^2}<0\,,
\]
provided that $\bb$ and $\nu$ satisfy the conditions in Lemma~\ref{lem-Rgs-B}.
\end{rem}

\begin{proof}[Proof of Lemma~\ref{lem-Rgs-B}]\

{\bf Lower bound:} Let $u_\bb\in \mathcal S_\ell$ be a minimizer of
$\mathcal G_{\bb,\nu;\ell}$. Since $u_\bb=0$ on
$(0,\infty)\times \partial K_\ell$, we can extend $u_\bb$ to a function
$\widetilde u_\bb$ by letting $\widetilde u_\bb=0$ in
$(0,\infty)\times(\R^2\setminus K_\ell)$. The function $\widetilde u_\bb$ is in
the form domain of the operator $\mathcal L(\nu)$ introduced
in~\eqref{eq-3D-op-2} and,
\begin{align*}
d(\bb,\nu;\ell)&=\mathcal G_{\bb,\nu;\ell}(u_\bb)\\
&=\int_{\R_+^3}\left(|\nabla-i\Eb_\nu)\widetilde u_\bb|^2
-\bb|\widetilde u_\bb|^2+\frac{\bb}{2}|\widetilde
u_\bb|^4\right)\,dx\,.
\end{align*}
Using the variational min-max principle, we get,
\begin{align*}
d(\bb,\nu;\ell)\geq  (\zeta(\nu)-\bb)\int_{\R^3}|\widetilde u_\bb|^2\,dx
=(\zeta(\nu)-\bb)\int_{U_\ell}|u_\bb|^2\,dx\,.
\end{align*}
By Theorem~\ref{thm-min-Rgl}, we have
$\displaystyle\int_{U_\ell}|u_\bb|^2\,dx\leq C_1\ell^2$, where $C_1$ is a
universal constant. If $\zeta(\nu)-\bb<0$, then
\[
d(\bb,\nu;\ell)\geq C_1(\zeta(\nu)-\bb)\ell^2\,.
\]
If $\zeta(\nu)-\bb\geq0$, then Lemma~\ref{lem-p-d(l)} tells us that
$d(\bb,\nu;\ell)=0$, thereby obtaining the lower bound announced in
Lemma~\ref{lem-Rgs-B}.

{\bf Upper bound:} We construct a test-configuration $f_\ell$ as follows.
Suppose $\nu\in(0,\pi/2)$. Let $\chi\in C_c^\infty(\R)$ be a cut-off function
such that $0\leq \chi\leq 1$, $\chi=1$ in $[-1/2,1/2]$ and $\chi=0$ in
$\R\setminus (-1,1)$. For all $v=(s,t)\in\R^2$, we  set
\begin{equation}\label{eq-eta}
\eta_\ell(v)=\chi_\ell(s)\chi_\ell(t)\\,\quad \chi_\ell(x)
=\chi\left(\frac{x}\ell\right)\quad (x\in\R)\,.\end{equation}
Let $\phi$ be the eigenfunction introduced in Remark~\ref{rem-2D-sp}, and let
$M>0$ be a given real number. For $j\in\mathbb{Z}$ we set
$c_j=Mj$. The number of such $c_j$ that belong to the interval $(-\ell,\ell)$
is given by $N(\ell,M)=2\{\ell/M\}+1$, where $\{r\}$ denotes the largest integer
strictly less than $r\in\R$. For $x=(x_1,x_2,x_3)\in\R^3$, we  define our trial
state $f_\ell(x)$ to be
\[
f_\ell (x)=\eta_\ell(x_2,x_3) \sum_{j=-\{\ell/M\}}^{\{\ell/M\}}
\exp\bigl(-i c_j \sin(\nu) x_3\bigr)\phi(x_1,x_2-c_j)\,.
\]
We compute $\mathcal G_{\bb,\nu;\ell}(t f_\ell)$, where $t>0$ is a constant
whose choice will be specified later. Actually,
\begin{multline}\label{eq:Gbig}
\mathcal G_{\bb,\nu;\ell}(tf_\ell)=
t^2\int_{U_{\ell}} \Bigl|\eta_\ell(x_2,x_3)\sum_{j=-\{\ell/M\}}^{\{\ell/M\}}
\exp\bigl(-i c_j \sin(\nu) x_3\bigr)(\partial_{x_1}\phi)(x_1,x_2-c_j)\Bigr|^2\\
+
\Bigl|\eta_\ell(x_2,x_3) \sum_{j=-\{\ell/M\}}^{\{\ell/M\}}
\exp\bigl(-i c_j \sin(\nu) x_3\bigr)(\partial_{x_2}\phi)(x_1,x_2-c_j)\\
+
\partial_{x_2}(\eta_\ell)(x_2,x_3) \sum_{j=-\{\ell/M\}}^{\{\ell/M\}}
\exp\bigl(-i c_j \sin(\nu) x_3\bigr)\phi(x_1,x_2-c_j)\Bigr|^2\\
+
\Bigl|\eta_\ell(x_2,x_3) \sum_{j=-\{\ell/M\}}^{\{\ell/M\}}
\exp\bigl(-i c_j \sin(\nu) x_3\bigr)\bigl(\cos(\nu)x_1+\sin(\nu)(x_2-c_j)\bigr)
\phi(x_1,x_2-c_j)\\
+\partial_{x_3}\bigl(\eta_\ell(x_2,x_3)\bigr) \sum_{j=-\{\ell/M\}}^{\{\ell/M\}}
\exp\bigl(-i c_j \sin(\nu) x_3\bigr)\phi(x_1,x_2-c_j)\Bigr|^2\\
-\bb\Bigl|\eta_\ell(x_2,x_3) \sum_{j=-\{\ell/M\}}^{\{\ell/M\}}
\exp\bigl(-i c_j \sin(\nu) x_3\bigr)\phi(x_1,x_2-c_j)\Bigr|^2\\
+\frac{\bb}{2}t^2\Bigl|\eta_\ell(x_2,x_3) \sum_{j=-\{\ell/M\}}^{\{\ell/M\}}
\exp\bigl(-i c_j \sin(\nu) x_3\bigr)\phi(x_1,x_2-c_j)\Bigr|^4
\,dx\,.
\end{multline}
We define $\mathcal G_{\text{main}}(tf_\ell)$ and
$\mathcal G_{\text{rest}}(tf_\ell)$ as follows,
\begin{multline*}
\mathcal G_{\text{main}}(tf_\ell)=
t^2\int_{U_{\ell}} \bigl|\eta_\ell(x_2,x_3)\bigr|^2
\sum_{j=-\{\ell/M\}}^{\{\ell/M\}}
\Bigl[\bigl|(\partial_{x_1}\phi)(x_1,x_2-c_j)\bigr|^2
+\bigl|(\partial_{x_2}\phi)(x_1,x_2-c_j)\bigr|^2\\
+\bigl|\bigl(\cos(\nu)x_1+\sin(\nu)(x_2-c_j)\bigr)\phi(x_1,x_2-c_j)\bigr|^2\\
-\bb|\phi(x_1,x_2-c_j)|^2+\frac{\bb}{2}t^2 |\phi(x_1,x_2-c_j)|^4\Bigr]\,dx
\end{multline*}
and
\begin{equation*}
\mathcal{G}_{\text{rest}}(tf_\ell)
= \mathcal G_{\bb,\nu;\ell}(tf_\ell)-\mathcal G_{\text{main}}(tf_\ell).
\end{equation*}
We will show that $\mathcal{G}_{\text{rest}}(tf_\ell)$ is small compared to
$\mathcal{G}_{\text{main}}(tf_\ell)$. We start by obtaining an upper bound on
$\mathcal{G}_{\text{main}}(tf_\ell)$. Using that
$\eta_\ell(x_2,x_3)=\chi_\ell(x_2)\chi_\ell(x_3)$ and
$0\leq \chi_\ell(x_2)\leq1$ we get for a fixed integer $j$
\begin{multline*}
t^2\int_{U_{\ell}} \bigl|\eta_\ell(x_2,x_3)\bigr|^2
\Bigl[\bigl|(\partial_{x_1}\phi)(x_1,x_2-c_j)\bigr|^2
+\bigl|(\partial_{x_2}\phi)(x_1,x_2-c_j)\bigr|^2\\
+\bigl|\bigl(\cos(\nu)x_1+\sin(\nu)(x_2-c_j)\bigr)\phi(x_1,x_2-c_j)\bigr|^2
+\frac{\bb}{2}t^2 |\phi(x_1,x_2-c_j)|^4\Bigr]\,dx\\
\leq
t^2 \int_{\mathbb{R}}|\chi_\ell(x_3)|^2\,dx_3
\int_{\mathbb{R}^2_+} \bigl|(\partial_{x_1}\phi)(x_1,x_2-c_j)\bigr|^2
+\bigl|(\partial_{x_2}\phi)(x_1,x_2-c_j)\bigr|^2\\
+\bigl|\bigl(\cos(\nu)x_1+\sin(\nu)(x_2-c_j)\bigr)\phi(x_1,x_2-c_j)\bigr|^2
+\frac{\bb}{2}t^2 |\phi(x_1,x_2-c_j)|^4 \,dx_1dx_2\\
=\ell\lambda t^2 \Bigl[\zeta(\nu)
+\frac{\bb}{2}t^2 \int_{\mathbb{R}^2_+}|\phi(x_1,x_2)|^4 \,dx_1dx_2\Bigr]
\end{multline*}
where $\lambda=\displaystyle\int_\R|\chi(z)|^2\,dz$. For the negative term, we
have
\begin{align*}
-\bb t^2 \int_{U_\ell} \bigl|\eta_\ell(x_2,x_3)\phi(x_1,x_2-c_j)|^2 &\, dx
= -\bb t^2 \int_{\mathbb{R}^3_+}
\bigl|\eta_\ell(x_2,x_3)\phi(x_1,x_2-c_j)|^2 \, dx\\
&= -\bb t^2 \int_{\mathbb{R}} |\chi_\ell(x_3)|^2\,dx_3
\int_{\mathbb{R}^2_+} \bigl|\chi_\ell(x_2)\phi(x_1,x_2-c_j)\bigr|^2 \, dx_1dx_2\\
&= -\bb \ell\lambda t^2 +{\bb}\ell\lambda t^2
\int_{\mathbb{R}^2_+}
(1-\bigl|\chi_\ell(x_2)\bigl|^2)|\phi(x_1,x_2-c_j)|^2 \, dx_1dx_2
\end{align*}
Using the exponential decay of $\phi$ we find that for all $s>0$
\[
-\bb t^2 \int_{U_\ell} \bigl|\eta_\ell(x_2,x_3)\phi(x_1,x_2-c_j)|^2 \, dx
= -\bb\ell\lambda t^2 + o(1/\ell^s)
\]
as $\ell\to\infty$. We conclude that
\[
\mathcal G_{\text{main}}(tf_\ell)
\leq N(\ell,M)\ell\lambda t^2\Bigl[\zeta(\nu)-\bb
+\frac{\bb}{2}t^2 \int_{\mathbb{R}^2_+}|\phi(x_1,x_2)|^4 \,dx_1dx_2\Bigr]
+N(\ell,M)o(1/\ell^s)
\]
as $\ell\to\infty$. We select $t$ sufficiently small so that,
\[
\zeta(\nu)-\bb+\frac{\bb t^2}2\int_{\R^2_+}|\phi|^4\,dx_1dx_2
\leq \frac{\zeta(\nu)-\bb}2,
\]
to get
\begin{equation}\label{eq:Gmain}
\mathcal G_{\text{main}}(tf_\ell)
\leq \frac{1}{2}N(\ell,M)\ell\lambda t^2(\zeta(\nu)-\bb)
+N(\ell,M)o(1/\ell^s)
\leq \frac{1}{2M}\ell^2\lambda t^2(\zeta(\nu)-\bb)
\end{equation}
as $\ell\to\infty$.

All terms in $\mathcal{G}_{\text{rest}}(tf_\ell)$ can be taken care of in the
same way, using~\eqref{eq:phidecay} (except for the term with
$\partial_{x_2}\eta_\ell(x_2,x_3)$ which is exponentially small due to the
decay of $\phi$). For this reason we only show how to handle the terms in
$\mathcal{G}_{\text{rest}}(tf_\ell)$ that comes from the first absolute value
in~\eqref{eq:Gbig}. We meet terms like
\begin{equation}\label{eq:typical}
t^2\int_{U_\ell}|\eta_\ell(x_2,x_3)|^2 e^{-ic_j
\sin(\nu)x_3}(\partial_{x_1}\phi)(x_1,x_2-c_j)
e^{i c_k\sin(\nu)x_3}(\partial_{x_1}\phi)(x_1,x_2-c_k)\,dx,
\end{equation}
where $j\neq k$ are such that $c_j$ and $c_k$ both belong to $(-\ell,\ell)$. We
assume that $j<k$ and estimate the absolute value of this integral by
\[
t^2\lambda\ell \int_{\R^2_+}
 |(\partial_{x_1})\phi(x_1,x_2-c_j)(\partial_{x_1})\phi(x_1,x_2-c_k)|\,dx_1dx_2.
\]
Next we decompose the integral into two terms, one where $x_2<(c_j+c_k)/2$ and
another one where $x_2>(c_j+c_k)/2$. Using~\eqref{eq:phidecay} we get
\begin{multline*}
\int_{\substack{(x_1,x_2)\in\R^2_+\\ x_2 < (c_j+c_k)/2}}
|(\partial_{x_1})\phi(x_1,x_2-c_j)(\partial_{x_1})\phi(x_1,x_2-c_k)|\,dx_1dx_2\\
=
\int_{\substack{(x_1,x_2)\in\R^2_+\\
x_2 < (c_j+c_k)/2}} e^{-\alpha/2\sqrt{x_1^2+(x_2-c_k)^2}}
e^{\alpha/2\sqrt{x_1^2+(x_2-c_k)^2}}|(\partial_{x_1})
\phi(x_1,x_2-c_j)(\partial_{x_1})\phi(x_1,x_2-c_k)|\,dx_1dx_2\\
\leq e^{-\alpha M (k-j)/4} \|\partial_{x_1}\phi\|_{L^2(\R^2_+)}
\bigl\|e^{\alpha/2\sqrt{x_1^2+(x_2)^2}}\partial_{x_1}\phi\bigr\|_{L^2(\R^2_+)}\\
\leq \sqrt{\zeta(\nu)}\sqrt{C_{\nu,\alpha}} e^{-\alpha M (k-j)/4}.
\end{multline*}
Here we have used the $L^\infty$ bound on the first exponential, and then the
Cauchy-Schwarz inequality together with the exponential decay on $\phi$
from~\eqref{eq:phidecay}. The same bound is true for the integral where
$x_2>(c_j+c_k)/2$. Looking at the case $j>k$ we get the same bound, but with
$j-k$ instead of $k-j$ in the exponential. We thus find
that~\eqref{eq:typical} in absolute value is bounded by
\begin{equation*}
2t^2\lambda \ell
\sqrt{\zeta(\nu)}\sqrt{C_{\nu,\alpha}}e^{-\frac{\alpha M}{4} |k-j|}.
\end{equation*}
We want to sum over all integers $j\neq k$ running from $-\{\ell/M\}$ to
$\{\ell/M\}$. For a real number $a<1$, the sum
\[
2\sum_{j=1}^{n-1} (n-j) a^j = \frac{2a}{(1-a)^2}\bigl(a^{n}-1+n(1-a)\bigr)
\leq \frac{2n a}{(1-a)^2}.
\]
Our sum transforms into this sum with $a=e^{-\alpha M/4}$ and
$n=\{\ell/M\}\leq \ell/M$ we find that the error coming from the first absolute
value in~\eqref{eq:Gbig} is bounded in absolute value by
\begin{equation*}
4t^2\lambda \ell^2 \sqrt{\zeta(\nu)}\sqrt{C_{\nu,\alpha}}
\frac{e^{-\alpha M/4}}{M(1-e^{-\alpha M/4})^2}.
\end{equation*}
Doing the same estimates for the other terms, we find that there exists a
$\beta>0$ and a constant $C>0$ (not depending on $M$)
such that
\begin{equation}\label{eq:Grest}
\mathcal{G}_{\text{rest}}(tf_\ell)\leq C t^2\lambda e^{-\beta M}\ell^2\,.
\end{equation}
Combining~\eqref{eq:Gmain} and~\eqref{eq:Grest} we find that
\[
\mathcal G_{\bb,\nu;\ell}(tf_\ell)
\leq \Bigl[\frac{1}{2M}(\zeta(\nu)-\bb) + C e^{-\beta M}\Bigr]\lambda t^2 \ell^2.
\]
Choosing $M$ large enough we find that there exist $\ell_0$ and $C_2>0$ such that
if $\ell\geq \ell_0$ then $\mathcal G_{\bb,\nu;\ell}(tf_\ell)\leq -C_2\ell^2$.
This finishes the proof for $\nu\in(0,\pi/2)$.

Suppose that $\nu=0$. In this specific case, we select the trial function
$f_\ell$ as follows,
\[
f_\ell(x)=\eta_\ell(x_2,x_3)\varphi_0(x_1)\,,
\]
where $\varphi_0$ is the ground state introduced in~\eqref{eq-1D-phi0}, and
$\eta_\ell$ is the cut-off function introduced  in~\eqref{eq-eta}. Consider
$t>0$. An easy computation yields,
\[
\mathcal G_{\bb,0;\ell}(tf_\ell)\leq \lambda^2t^2\ell^2
\left(\Theta_0-\bb+\frac{\bb t^2}{2}\int_{\R_+}|\varphi_0(x_1)|^4\,dx_1\right)
+Ct^2\,,
\]
for some constant $C$. Selecting $t$ sufficiently small so that
$\Theta_0-\bb\displaystyle
+\frac{\bb t^2}{2}\int_{\R_+}|\varphi_0(x_1)|^4\,dx_1<\frac{\Theta_0-\bb}2$,
we get the upper bound that we want to prove.
\end{proof}

\subsubsection{The thermodynamic limit}\label{sec-E(nu)}

The aim of this section is to prove the following theorem.

\begin{thm}\label{thm-thdL}\
\begin{enumerate}
\item
Suppose $\bb\in[\Theta_0,1]$ and $\nu\in[0,\pi/2]$. There exists a
constant $E(\bb,\nu)$ such that
\[
\liminf_{\ell\to\infty}\frac{d(\bb,\nu;\ell)}{4\ell^2}
=\limsup_{\ell\to\infty}\frac{d(\bb,\nu;\ell)}{4\ell^2}=E(\bb,\nu)\,.
\]
Furthermore, $E(\bb,\nu)<0$ if $\bb>\zeta(\nu)$ and $E(\bb,\nu)=0$ otherwise.
\item There exist positive universal  constants $\ell_0$ and $C$ such that,
\begin{equation}\label{eq-est-tdl}
E(\bb,\nu)\leq \frac{d(\bb,\nu;\ell)}{4\ell^2}
\leq  E(\bb,\nu)+\frac{C}{\ell^{2/3}}\,,\quad\forall~\nu\in[0,\pi/2]\,,
\quad\forall~\ell\geq\ell_0\,.
\end{equation}
\end{enumerate}
\end{thm}

The proof of Theorem~\ref{thm-thdL} relies on the following abstract lemma,
which is proved in~\cite{FK3D}.

\begin{lem}\label{gen-lem}
Consider a decreasing function $d:(0,\infty)\to(-\infty,0]$ such that the
function
$f:(0,\infty)\ni\ell\mapsto \frac{d(\ell)}{\ell^2}\in\R$
is bounded.

Suppose that there exist constants $C>0$ and $\ell_0>0$ such that  the estimate
\begin{equation}\label{eq-gen-lem-ass}
f(n\ell) \geq f((1+a)\ell)-C\left(a+\frac1{a^2\ell^2}\right)\,,\end{equation}
holds true for all $a\in(0,1)$,
$n\in\mathbb N$ and $\ell\geq\ell_0$.

Then $f(\ell)$ has a limit $A$ as $\ell\to\infty$. Furthermore, for all
$\ell\geq 2\ell_0$, the following estimate holds true,
\begin{equation}\label{eq-gen-lem}
f(\ell)\leq A+\frac{2C}{\ell^{2/3}}\,.
\end{equation}
\end{lem}

In order to use the result of Lemma~\ref{gen-lem} for the energy
$d(\bb,\nu,\ell)$, we establish the estimate in  Lemma~\ref{lem-thdL} below.

\begin{lem}\label{lem-thdL}
Let $\bb\in(\Theta_0,1)$ and $\nu\in(0,\pi/2)$ such that $\bb\geq\zeta(\nu)$.
There exist universal constants $C>0$ and $\ell_0\geq 1$ such that, for all
$\ell\geq \ell_0$,  $n\in\mathbb N$ and $a\in(0,1)$, we have,
\[
\frac{d(\bb,\nu;n\ell)}{(n\ell)^2}
\geq \frac{d(\bb,\nu;(1+a)\ell)}{\ell^2}-\frac{C}{a^2\ell^2}\,.
\]
\end{lem}
\begin{proof}
Let $n\geq 2$ be a natural number. If $a\in(0,1)$ and $j=(j_1,j_2)\in\Z^2$,
we denote by
\[
K_{a,j}=I_{j_1}\times I_{j_2}\,,
\]
where
\[
\forall~m\in\mathbb Z\,,
\quad I_m=\bigg{(}m-n-1-a-(1+a)\,,\,m-n-1-a+(1+a)\bigg{)}\,.
\]
Consider a partition of unity $(\chi_j)$ of $\R^2$ such that:
\[
\sum_{j}|\chi_j|^2=1\,,\quad 0\leq\chi_j\leq1\quad\text{ in }\R^2\,,
\quad \supp\chi_j\subset K_{a,j}\,,\quad
|\nabla\chi_j|\leq \frac{C}{a}\,,
\]
where $C$ is a universal constant.
We define $\chi_{\ell,j}(x)=\chi_j(x/\ell)$. Then we obtain a new partition of
unity $\chi_{\ell,j}$ such that
$\supp\chi_{\ell,j}\subset \mathcal K_{\ell,j}$,
with
\[
\mathcal  K_{\ell,j}=(\ell x~:~x\in K_{a,j}\}\,.
\]
Let $\mathcal J=\{j=(j_1,j_2)\in \Z^2~:~1\leq  j_1,j_2\leq n\}$ and
$K_{n\ell}=(-n\ell,n\ell)\times(-n\ell,n\ell)$.
Then the family
$(\mathcal K_{\ell,j})_{j\in\mathcal J}$ is a covering of  $K_{n\ell}$, and is
formed exactly of $n^2$ squares.

We  restrict the partition of unity $(\chi_{\ell,j})$ to the set
$K_{n\ell}=(-n\ell,n\ell)\times(-n\ell,n\ell)$.
Let $u_{n\ell}$ be a minimizer of~\eqref{eq-Rgl}, i.e.
$\mathcal G_{\bb,\nu;n\ell}(u_{n\ell})=d(\bb,\nu;n\ell)$.
We have the following decomposition formula,
\begin{equation}\label{eq-IMS}
\mathcal G_{\bb,\nu;n\ell}(u_{n\ell})\geq \sum_{j\in\mathcal J}
\left(\mathcal G_{\bb,\nu;n\ell}(\chi_{\ell,j}u_{n\ell})
-\|\,|\nabla\chi_{\ell,j}|\,u_{n\ell}\|^2_{L^2(K_{n\ell})}\right)\,.
\end{equation}
By magnetic translation invariance, we get
$\mathcal G_{\bb,\nu;n\ell}(\chi_{\ell,j}u_{n\ell})\geq d(\bb,\nu;(1+a)\ell)$.
Therefore, it results from~\eqref{eq-IMS},
\begin{equation}\label{eq-lb-IMS}
d(\bb,\nu;n\ell)
\geq n^2d(\bb,\nu;(1+a)\ell)-\sum_j\|
\nabla\chi_{\ell,j}\,u_{n\ell}\|^2_{L^2(K_{n\ell})}\,.
\end{equation}
 Using Theorem~\ref{thm-min-Rgl}, we get further,
\begin{equation}\label{eq-lb1-IMS}
d(\bb,\nu;n\ell)\geq n^2d(\bb,\nu;(1+a)\ell)-\frac{Cn^2}{a^2}\,.
\end{equation}
Therefore,
\[
\frac{d(\bb,\nu;n\ell)}{(n\ell)^2}\geq \frac{d(\bb,\nu;(1+a)\ell)}{\ell^2}
-\frac{C}{a^2\ell^2}\,.
\]
\end{proof}

\begin{rem}\label{rem-thdL}
Using the lower bound in Lemma~\ref{lem-Rgs-B}, we infer from
Lemma~\ref{lem-thdL} the following lower bound,
\[
\frac{d(\bb,\nu;n\ell)}{(n\ell)^2}
\geq \frac{d(\bb,\nu;(1+a)\ell)}{((1+a)\ell)^2}
-C\left(a+\frac1{a^2\ell^2}\right)\,,
\]
where the constant $C>0$ is universal.
\end{rem}

\begin{proof}[Proof of Theorem~\ref{thm-thdL}]\

Let $f(\ell)=\displaystyle\frac{d(\bb,\nu;\ell)}{\ell^2}$. Thanks to
Lemmas~\ref{lem-Rgs-B} and~\ref{lem-thdL}, we know that the functions $f(\ell)$
and $d(\ell):=d(\bb,\nu;\ell)$ satisfy the assumptions in Lemma~\ref{gen-lem}.
Consequently, $f(\ell)$ has a limit $E(\bb,\nu;\ell)$ as $\ell\to\infty$ and,
\[
f(\ell)\leq E(\bb,\nu)+\frac{C}{\ell^{2/3}}\,,
\]
for all $\ell\geq\ell_0$. Here $C$ and $\ell_0$ are constants independent of
$\bb\in[\Theta_0,1]$ and $\nu\in[0,\pi/2]$.

It remains to establish a lower bound of $f(\ell)$.
Let $n\in\mathbb N$. By using a comparison argument and magnetic translation
invariance, we have that
\[
d(\bb,\nu;n\ell)\leq n^2 d(\bb,\nu;\ell)\,.
\]
Consequently,
\[
\frac{d(\bb,\nu;\ell)}{\ell^2}\geq \frac{d(\bb,\nu;n\ell)}{(n\ell)^2}\,.
\]
Making $n\to\infty$ in both sides above, we conclude,
\begin{equation}\label{eq-est-d(ell)}
\frac{d(\bb,\nu;\ell)}{\ell^2}\geq E(\bb,\nu)\,.
\end{equation}
Lemma~\ref{lem-Rgs-B} tells us  that if $\bb\geq\zeta(\nu)$, then
$d(\bb,\nu;\ell)<0$. In this case, we get  $E(\bb,\nu)<0$ as consequence
of~\eqref{eq-est-d(ell)}. On the other hand, If $\bb\leq\zeta(\nu)$, we know by
Lemma~\ref{lem-p-d(l)} that $E(\bb,\nu)=0$.
\end{proof}

\subsubsection{Properties of the function $E(\bb,\nu)$}

Theorem~\ref{thm-thdL} provides us with a limiting constant
$E(\bb,\nu)\in(-\infty,0]$ defined for $\bb\in[\Theta_0,1]$ and
$\nu\in[0,\pi/2]$. In this section, we will study properties of $E(\bb,\nu)$
as a function of $\bb$ and $\nu$.

\begin{thm}\label{thm-E(b,nu)}
For all $(\bb,\nu)\in [\Theta_0,1]\times[0,\pi/2]$, let the constant
$E(\bb,\nu)$ be defined as in Theorem~\ref{thm-thdL}.
Then $E(\bb,\nu)$ satisfies the following properties.
\begin{enumerate}
\item Given $\nu\in[0,\pi/2]$, the function
$[\Theta_0,1]\ni \bb\mapsto E(\bb,\nu)$ is continuous and monotone decreasing.
\item Given $\bb\in[\Theta_0,1]$, the function
$[0,\pi/2]\ni\nu\mapsto E(\bb,\nu)$ is continuous.
\end{enumerate}
\end{thm}
\begin{proof}

{\bf Continuity and monotonicity of $\bb\mapsto E(\bb,\nu)$:}

Let $\bb\in[\Theta_0,1]$ and $\varepsilon\in\R$ such that
$\bb+\varepsilon\in[\Theta_0,1]$. Recall the definition of the functional
$\mathcal G_{\bb,\nu;\ell}$ in~\eqref{eq-Rgl} together with the associated
ground state energy $d(\bb,\nu;\ell)$. It is easy to check that,
\begin{equation}\label{eq-mon}
\mathcal G_{\bb+\varepsilon,\nu;\ell}(u)
=\mathcal G_{\bb,\nu;\ell}(u)+\varepsilon\int_{U_\ell}\left(-|u|^2
+\frac12|u|^4\right)\,dx\,,
\end{equation}
valid for any $u\in\mathcal S_\ell$. In particular, setting successively
$u=\varphi_{\bb,\nu;\ell}$ then $u=\varphi_{\bb+\varepsilon,\nu;\ell}$
in~\eqref{eq-mon} and using the properties in Theorem~\ref{thm-min-Rgl}, we get
the following estimate,
\[
|d(\bb+\varepsilon,\nu;\ell)-d(\bb,\nu;\ell)|\leq C\varepsilon\ell^2\,,
\]
for some universal constant $C$. Remembering the definition of $E(\cdot,\nu)$,
we get,
\[
|E(\bb+\varepsilon,\nu)-E(\bb,\nu)|\leq C\varepsilon\,,
\]
thereby proving the continuity   of $E(\cdot,\nu)$.

To obtain monotonicity of $E(\bb,\nu)$, we suppose that $\varepsilon<0$ and we
set $u=\varphi_{\bb+\varepsilon,\nu;\ell}$ in~\eqref{eq-mon}. Thanks to
Theorem~\ref{thm-min-Rgl}, we know that $|u|\leq 1$ and consequently,
\begin{align*}
d(\bb+\varepsilon,\nu;\ell)&
\geq d(\bb,\nu;\ell)-\frac{\varepsilon}2\int_{U_\ell}|u|^2\,dx\\
&\geq d(\bb,\nu;\ell)\,.\end{align*}
Dividing both sides above by $4\ell^2$ then letting $\ell\to\infty$, we get,
\[
E(\bb+\varepsilon,\nu;\ell)\geq E(\bb,\nu;\ell)\,,
\]
which proves that $E(\cdot,\nu)$ is monotone decreasing.

{\bf Continuity  of $\nu\mapsto E(\bb,\nu)$:}

Let  $\nu\in[0,\pi/2]$ and $\varepsilon\in(-1,1)\setminus\{0\}$ such that
$\nu+\varepsilon\in[0,\pi/2]$.  We want to prove that
$E(\bb,\nu+\varepsilon)\to E(\bb,\nu)$ as $\varepsilon\to0$.

Recall the definition of the magnetic potential $\Eb_\nu$ in~\eqref{eq-3D-Eb}.
Notice that, if $u\in\mathcal S_\ell$, we have the following estimate,
\begin{equation}\label{eq-E(nu)}
\left|\mathcal G_{\bb,\nu+\varepsilon;\ell}(u)
-\mathcal G_{\bb,\nu;\ell}(u)\right|
\leq |\varepsilon|\int_{U_\ell}|(\nabla-i\Eb_\nu)u|^2\,dx
+2|\varepsilon^{-1}|\int_{U_\ell}|(\Eb_{\nu+\varepsilon}-\Eb_\nu)u|^2\,dx\,.
\end{equation}

Using the bounds,
\[
|\cos(\nu+\varepsilon)-\cos\nu|\leq\varepsilon\,,\quad |\sin(\nu+\varepsilon)
-\sin\nu|\leq \varepsilon\,,
\]
we get,
\[
|\Eb_{\nu+\varepsilon}(x)-\Eb_\nu(x)|\leq \varepsilon\left(|x_1|+|x_2|\right),
\quad\forall~x=(x_1,x_2,x_3)\in\R^3\,.
\]
Consequently, we infer from~\eqref{eq-E(nu)} the following bound,
\begin{equation}\label{eq-E(nu)'}
\left|\mathcal G_{\bb,\nu+\varepsilon;\ell}(u)-\mathcal G_{\bb,\nu;\ell}(u)\right|
\leq |\varepsilon|\int_{U_\ell}|(\nabla-i\Eb_\nu)u|^2\,dx
+2|\varepsilon|\int_{U_\ell}\left(|x_1|^2+|x_2|^2\right)|u|^2\,dx\,.
\end{equation}
Let $\zeta\in C_c^\infty(\R)$ be a cut-off function satisfying
$0\leq \zeta\leq 1$ in $\R$, $\supp\zeta\subset[-1,1]$ and
$\zeta=1$ in $[-\frac12,\frac12]$. Let $\zeta_\ell(x)=\zeta(x_1/\ell)$ and
$\varphi_{\bb,\nu+\varepsilon;\ell}\in\mathcal S_\ell$ be a minimizer of
$\mathcal G_{\bb,\nu+\varepsilon;\ell}$ given by Theorem~\ref{thm-min-Rgl}.
Applying~\eqref{eq-E(nu)'} with $u=\zeta_\ell\varphi_\ell$ and using the decay
estimate of Theorem~\ref{thm-min-Rgl}, we get for all $\ell\geq 1$,
\begin{equation}\label{eq-E(nu)''}
\mathcal G_{\bb,\nu+\varepsilon;\ell}(\zeta_\ell\varphi_{\bb,\nu+\varepsilon;\ell})
\geq d(\bb,\nu;\ell)-C|\varepsilon|\ell^4\,,
\end{equation}
where $C>0$ is a universal constant.  Using the estimate in~\eqref{eq-est-tdl},
we get  for all $\ell\geq \ell_0$,
\begin{equation}\label{eq-E(nu)'''}
\mathcal G_{\bb,\nu+\varepsilon;\ell}(\zeta_\ell\varphi_{\bb,\nu+\varepsilon;\ell})
\geq (4\ell^2)E(\bb,\nu)-C|\varepsilon|\ell^4\,,
\end{equation}
where $\ell_0$ is a universal constant.

We estimate the term
$\mathcal G_{\bb,\nu+\varepsilon;\ell}(\zeta_\ell\varphi_{\bb,\nu+\varepsilon;\ell})$
from above. Actually, an integration by parts give us the following identity,
\[
\mathcal G_{\bb,\nu+\varepsilon;\ell}(\zeta_\ell\varphi_{\bb,\nu+\varepsilon;\ell})
=\int_{U_\ell}|\zeta'_\ell|^2|\varphi_{\bb,\nu+\varepsilon;\ell}|^2\,dx
+\int_{U_\ell}\zeta^2_\ell\left(\frac12\zeta^2_\ell-1\right)
|\varphi_{\bb,\nu+\varepsilon;\ell}|^4\,dx\,.
\]
Using the decay estimate of Theorem~\ref{thm-min-Rgl}, we get,
\[
\mathcal G_{\bb,\nu+\varepsilon;\ell}(\zeta_\ell\varphi_{\bb,\nu+\varepsilon;\ell})
\leq C
-\frac12\int_{U_\ell}|\varphi_{\bb,\nu+\varepsilon;\ell}|^4\,dx\,.
\]
The Ginzburg-Landau equation satisfied by $\varphi_{\bb,\nu+\varepsilon;\ell}$
gives us $$\mathcal G_{\bb,\nu+\varepsilon;\ell}(\varphi_{\bb,\nu+\varepsilon;\ell})
=-\frac12\int_{U_\ell}|\varphi_{\bb,\nu+\varepsilon;\ell}|^4\,dx.$$
Consequently, we get the upper bound,
\[
\mathcal G_{\bb,\nu+\varepsilon;\ell}(\zeta_\ell\varphi_{\bb,\nu+\varepsilon;\ell})
\leq d(\bb,\nu+\varepsilon;\ell)+C\,.
\]
Using the estimate in~\eqref{eq-est-tdl}, we get further,
\[
\mathcal G_{\bb,\nu+\varepsilon;\ell}(\zeta_\ell\varphi_{\bb,\nu+\varepsilon;\ell})
\leq (4\ell^2)E(\bb,\nu+\varepsilon)+C\ell^{4/3}\,.
\]
Inserting this upper bound into~\eqref{eq-E(nu)'''}, we get,
\[
E(\bb,\nu+\varepsilon)\geq E(\bb,\nu)-C|\varepsilon|\ell^4-C\ell^{-2/3}\,.
\]
Taking successively $\displaystyle\liminf_{\varepsilon\to 0}$ then
$\displaystyle\lim_{\ell\to\infty}$ on both sides above, we get,
\[
\liminf_{\varepsilon\to0}E(\bb,\nu+\varepsilon)\geq E(\bb,\nu)\,.
\]
In a similar fashion,  by applying~\eqref{eq-E(nu)} with
$u=\zeta_\ell\varphi_{\bb,\nu;\ell}$ and $\varphi_{\bb,\nu;\ell}$ being a
minimizer of $\mathcal G_{\bb,\nu;\ell}$, we can prove that
$\displaystyle\limsup_{\varepsilon\to0}E(\bb,\nu+\varepsilon)\leq E(\bb,\nu)$.
This gives us that
\[
\lim_{\varepsilon\to0}E(\bb,\nu+\varepsilon)=E(\bb,\nu)\,,
\]
and thereby proves the continuity of the function $E(\bb,\cdot)$.
\end{proof}

\section{Auxiliary results}
We collect some results that are necessary for controlling the errors resulting
from various approximations. The  first result is an estimate obtained
in \cite[Lemma~3.2 and Theorem~3.3]{Al} (for a different method
see~\cite[Theorem~12.3.1]{FH-b}).

\begin{lemma}\label{lem:almog}
There exists a constant $C_1>0$ such that if
$(\psi,\Ab)\in H^1(\Omega)\times \dot H^1_{\Div,\Fb}(\R^3)$ is a solution
to~\eqref{eq-3D-GLeq}, then
\[
\|\psi\|_{L^4(\Omega)}^4\leq C_1\lambda.
\]
Here,
\begin{equation}\label{eq-zeta-2}
\lambda=\max\left(\frac1\kappa,\left[\frac{\kappa}H-1\right]_+^2\right)\,.
\end{equation}
\end{lemma}

The next lemma is taken from \cite[Lemma~10.33]{FH-b}, which, together with
Lemma~\ref{lem:almog}, give a good estimate of $\|\curl(\Ab-\Fb)\|_{L^2(\R^3)}$.

\begin{lem}\label{lem-curl}
There exists a constant $C>0$ such that, if $(\psi,\Ab)\in H^1(\Omega;\C)
\times \dot H^1_{\Div,\Fb}(\R^3)$ is a solution of~\eqref{eq-3D-GLeq}, then,
\begin{equation}\label{eq-est-curl}
\|\curl(\Ab-\Fb)\|_{L^2(\R^3)}\leq \frac{C}{H}\|\psi\|_{L^4(\Omega)}^2\,,
\end{equation}
for all $\kappa>0$ and $H>0$.
\end{lem}

The next result is Theorem~\ref{thm-3D-ad-est} below which is proved
in~\cite{FK3D}. Similar estimates to those in Theorem~\ref{thm-3D-ad-est} are
also given in~\cite{Pa}.

\begin{thm}\label{thm-3D-ad-est}
Suppose that $0<\Lambda_{\min}\leq \Lambda_{\max}$. There exist constants
$\kappa_0>1$ and $C_1>0$ such that, if
\[
\kappa\geq\kappa_0\,,\quad \Lambda_{\min}\leq\frac\kappa{H}
\leq \Lambda_{\max}\,,
\]
and $(\psi,\Ab)\in H^1(\Omega;\C)\times \dot H^1_{\Div,\Fb}(\R^3)$
is a solution of~\eqref{eq-3D-GLeq}, then
\begin{align}
&\|(\nabla-i\kappa H\Ab)\psi\|_{C(\overline{\Omega})}
\leq C_1\sqrt{\kappa H}\|\psi\|_{L^\infty(\Omega)}\,,\label{eq-est1}\\
&
\|\Ab-\Fb\|_{W^{2,6}(\Omega)}\leq
C_1\left(\|\curl(\Ab-\Fb)\|_{L^2(\R^3)}+\frac1{\sqrt{\kappa H}}
\|\psi\|_{L^6(\Omega)}\|\psi\|_{L^\infty(\Omega)}
\right),\label{eq-est2'}\\
&
\|\Ab-\Fb\|_{C^{1,1/2}(\Omega)}\leq C_1\left(\|\curl(\Ab-\Fb)\|_{L^2(\R^3)}
+\frac1{\sqrt{\kappa H}}
\|\psi\|_{L^6(\Omega)}\|\psi\|_{L^\infty(\Omega)}
\right),\label{eq-est2}\\
&\|\curl(\Ab-\Fb)\|_{C^{0,1/2}(\overline{\Omega})}\leq C_1
\left(\|\curl(\Ab-\Fb)\|_{L^2(\R^3)}
+\frac1{\sqrt{\kappa H}}\|\psi\|_{L^6(\Omega)}
\|\psi\|_{L^\infty(\Omega)}\right). \label{eq-est2*}
\end{align}
\end{thm}

Combining the results in Lemma~\ref{lem:almog}, Lemma~\ref{lem-curl} and
Theorem~\ref{thm-3D-ad-est}, we get the following corollary.

\begin{corol}\label{corol:L2}
Suppose that $0<\Lambda_{\min}\leq \Lambda_{\max}$. There exist constants
$\kappa_0>1$ and $C_1>0$ such that, if
\[
\kappa\geq\kappa_0\,,\quad \Lambda_{\min}\leq\frac\kappa{H}
\leq \Lambda_{\max}\,,
\]
and $(\psi,\Ab)\in H^1(\Omega;\C)\times \dot H^1_{\Div,\Fb}(\R^3)$
is a solution of~\eqref{eq-3D-GLeq}, then,
\[
\|\Ab-\Fb\|_{C^{1,1/2}(\Omega)}\leq C_1\frac{\lambda^{1/6}}{\kappa}\,.
\]
\end{corol}

\section{Lower bound}\label{sec:lb}

In this section we will prove  Theorem~\ref{thm-lb} below, whose statement
requires some notation. Let $D\subset\Omega$ be a given  open set such that
there exists a subset $\widetilde D$ of $\R^3$ having smooth boundary and
$D=\widetilde D\cap\Omega$. For all $a>0$, we assign to $D$ the following subset
of $\Omega$,
\begin{equation}\label{eq-Da}
D_a=\{x\in \Omega~:~\dist(x,D)\leq a\}\,.
\end{equation}

We introduce the following functional,
\begin{equation}\label{eq-GLe0}
\mathcal E_0(u,\Ab;D)
=\int_{D}\left(|(\nabla-i\Ab)u|^2-\kappa^2|u|^2+\frac{\kappa^2}2|u|^4\right)\,dx\,,
\end{equation}
where $u\in H^1(\Omega;\C)$ and $\Ab\in \dot H^1_{\Div,\Fb}(\R^3)$. If
$D=\Omega$, we omit the dependence on the domain and write
$\mathcal E_0(\psi,\Ab)$ for $\mathcal E_0(\psi,\Ab;\Omega)$.

\begin{thm}\label{thm-lb}
Suppose that the magnetic field $H$ is a function of $\kappa$ such that,
\[
1\leq \liminf_{\kappa\to\infty}\frac{H}{\kappa}
\leq\limsup_{\kappa\to\infty}\frac{H}{\kappa}<\infty\,.
\]
Let $\kappa\ni\R_+\mapsto a(\kappa)\in\R_+$ be a function satisfying
$\displaystyle\lim_{\kappa\to\infty}a(\kappa)=0$. Then, for any solution
$(\psi,\Ab)\in H^1(\Omega;\C)\times \dot H^1_{\Div,\Fb}(\R^3)$
of~\eqref{eq-3D-GLf} and any function $h\in C^1(\Omega)$ satisfying
$\|h\|_{L^\infty(\Omega)}\leq1$ and $\supp h\subset \overline{D_{a(\kappa)}}$, the
following asymptotic lower bound holds,
\begin{multline}\label{eq-3D-hc2-lb}
\mathcal E_0(h\psi,\Ab)
\geq \sqrt{\kappa H}\int_{\overline D\cap\partial\Omega}
E\left(\frac{\kappa}H,\nu(x)\right)\,d\sigma(x)+E_2|D|\,[\kappa-H]_+^2
+o\bigg(\max\left(\kappa,[\kappa-H]_+^2\right)\bigg)\,\\
\text{as}\quad\kappa\to\infty\,.
\end{multline}
Here $d\sigma(x)$ is the surface measure on the boundary of $\Omega$,
$E_2<0$ is the universal constant introduced in~\eqref{eq-E2}, and
$\mathcal E_0$ is the functional introduced in~\eqref{eq-GLe0}.
\end{thm}

The proof of Theorem~\ref{thm-lb} is a
bit lengthy and is divided into several subsections. We will introduce two
parameters
\[
\alpha=\alpha(\kappa)\in(0,1)\,,\quad
\delta=\delta(\kappa)\in(0,1)\,,
\]
that will be  used along the proof. Different conditions on these
parameters will arise so as to control the error terms correctly.
The choice of the parameters will be fixed at the end of the proof
(in~\eqref{eq:deltaalpha} below) so that they satisfy the
aforementioned conditions, and are negligible in the limit of
large $\kappa$.

Throughout the section, $(\psi,\Ab)\in H^1(\Omega)\times \dot
H^1_{\Div,\Fb}(\R^3)$ will always denote a solution to~\eqref{eq-3D-GLeq}.

\subsection{Splitting into bulk and surface terms}
We introduce smooth real-valued functions $\chi_1$ and $\chi_2$ such that
$\chi_1^2+\chi_2^2= 1$ in $\Omega$,
\[
\chi_1(x)=
\begin{cases}
1, &\text{if } \dist(x,\partial\Omega)<\delta/2,\\
0, &\text{if } \dist(x,\partial\Omega)>\delta,
\end{cases}
\]
and $|\nabla \chi_j|\leq C/\delta$ for $j=1,2$ and some constant $C$
independent of $\delta$. Using the IMS decomposition formula and the fact that
$\int_{\Omega} (\chi_j(x)^2-\chi_j(x)^4)|\psi|^4\,dx\geq 0$ (since
$0\leq\chi_j(x)\leq 1$), we find that
\begin{equation}\label{eq-splitting}
\mathcal E_0(h\psi,\Ab)
\geq \mathcal E_0(\chi_1h\psi,\Ab)+\mathcal E_0(\chi_2h\psi,\Ab)
-\sum_{j=1}^2\int_{\Omega} |\nabla\chi_j|^2|h\psi|^2\,dx.
\end{equation}
To estimate the error we use that $\| \psi \|_{\infty} \leq 1$ and the fact that
the measure of the support of $\nabla\chi_j$ is bounded by $C\delta$ for some
constant $C$, so
get
\begin{equation}\label{eq-B+S}
\sum_{j=1}^2\int_{\Omega} |\nabla\chi_j|^2|h\psi|^2\,dx
\leq C\delta^{-1}\,.
\end{equation}

Next, we estimate separately the terms $\mathcal E_0(\chi_1h\psi,\Ab)$ (surface
energy) and $\mathcal E_0(\chi_2h\psi,\Ab)$ (bulk energy).

\subsection{The surface energy}

The estimate of the surface energy requires two steps, a decomposition of the
energy via a partition of unity, then passing to local boundary coordinates
that allow us to compare with the model case of the half-space that is studied
in Section~\ref{sec:hsp}.

\subsubsection{Boundary coordinates}\label{sec:bndcod}
We introduce a system of coordinates valid near a point of the boundary. These
coordinates are used in \cite{hemo} and then in \cite{R} in order to estimate
the ground state energy of a magnetic Schr\"odinger operator with large magnetic
field (or with small semi-classical parameter).

Consider a point  $x_0\in\partial\Omega$. After performing a translation, we
may assume that the Cartesian coordinates of $x_0$ are all $0$, i.e. $x_0=0$.
Let $V$ be a neighborhood of $x_0$ such that there
exists local boundary coordinates $(y_2,y_3)$ in $W=V\cap \partial\Omega$, i.e.
there exists an open subset $U$ of $\R^2$ and a diffeomorphism $\phi:W\to U$,
$\phi(x)=(y_2,y_3)$. Denote by $N$ the inward pointing
normal at the point $\phi^{-1}(y_2,y_3)\in\partial\Omega$. We define the
coordinate transformation $\Phi$ as
\[
(x_1,x_2,x_3)=\Phi^{-1}(y_1,y_2,y_3)=\phi^{-1}(y_2,y_3)+y_1 N.
\]
The standard Euclidean metric $g_0=\sum_{j=1}^3 dx_j \otimes dx_j $ transforms
to
\begin{align*}
g_0 &= \sum_{1\leq j,k\leq 3} g_{jk} dy_j \otimes dy_k \\
    &= dy_3\otimes dy_3+\sum_{2\leq j,k\leq 3} \Bigl[
       G_{jk}(y_2,y_3)
       -2y_1 K_{jk}(y_2,y_3)
       +y_1^2 L_{jk}(y_2,y_3)
       \Bigr]
       dy_j\otimes dy_k\,
\end{align*}
where
\begin{align*}
G&=\sum_{2\leq k,j\leq 3}G_{jk}\, dy_j\otimes dy_k
 = \sum_{\substack{2\leq k,j\leq 3\\1\leq l\leq 3}}
         \Bigl\langle\frac{\partial x_l}{\partial y_j}
       ,\frac{\partial x_l}{\partial y_k}\Bigr\rangle\, dy_j\otimes dy_k\,,\\
K&=\sum_{2\leq k,j\leq 3}K_{jk}\, dy_j\otimes dy_k
 = \sum_{2\leq k,j\leq 3}-\Bigl\langle \frac{\partial N}{\partial y_j}
       ,\frac{\partial x}{\partial y_k}\Bigr\rangle\, dy_j\otimes dy_k\,
        \ \text{and} \\
L&=\sum_{2\leq k,j\leq 3}L_{jk}\, dy_j\otimes dy_k
 = \sum_{2\leq k,j\leq 3}\Bigl\langle \frac{\partial N}{\partial y_j}
       ,\frac{\partial N}{\partial y_k}\Bigr\rangle\, dy_j\otimes dy_k
\end{align*}
are the first, second and third fundamental forms on $\partial\Omega$.
We denote by $g^{jk}$ the inverse of $g_{jk}$. Its Taylor expansion, valid in
the neighborhood, is given by
\begin{equation}\label{eq:metricapprox}
(g^{jk})_{1\leq j,k\leq 3} =
Id
+
\begin{pmatrix}
0 & 0 & 0\\
0 & \mathcal O(|y|) & \mathcal O(|y|)\\
0 & \mathcal O(|y|) & \mathcal O(|y|)
\end{pmatrix}.
\end{equation}
The Lebesgue measure $dx$ transforms into $dx=\det(g_{jk})^{1/2}dy$. This
determinant has the Taylor expansion
\begin{equation}\label{eq:jacobian}
\det(g_{jk})^{1/2} = 1 + \mathcal O(|y|),
\end{equation}
valid in the neighborhood. The magnetic vector potential
$\Fb=(F_1,F_2,F_3)=(-x_2/2,x_1/2,0)$ is transformed to
a magnetic potential
$\widetilde{\Fb}=(\widetilde{F}_1,\widetilde{F}_2,\widetilde{F}_3)$ given as
follows,
\begin{equation}\label{eq-gauge-F}
\widetilde{F}_j(y)
 = \sum_{k=1}^3 F_k(\Phi^{-1}(y))\frac{\partial x_k}{\partial y_j}.
\end{equation}
In~\cite{R}, a particular choice of a gauge transformation is selected so that,
in a neighborhood of the point $x_0$,  the new vector potential
$\widetilde{\mathbf{F}}$ satisfies,
\begin{align}\label{eq:tF}
\widetilde{F}_1 &= 0,\qquad&
\widetilde{F}_2 &= \mathcal O(|y|^2),\qquad&
\widetilde{F}_3 &= y_1\cos\nu+y_2\sin\nu+ \mathcal O(|y|^2),
\end{align}
Here $\nu=\nu(x_0)$ is the angle between
the magnetic field and the tangent plane of $\partial\Omega$ at the point $x_0$.

Notice that the constants implicit in the ${\mathcal O}$ notation in \eqref{eq:metricapprox}, \eqref{eq:jacobian} and \eqref{eq:tF} can be chosen uniform (i.e. independent of the boundary point $x_0$) by compactness and regularity of $\partial \Omega$.

If $u$ is a function  with support
in a coordinate neighborhood $V$, we may express the functional
$\mathcal E_0(u,\Fb)$  explicitly in the new coordinates as follows,
\begin{multline}
\mathcal E_0^{\rm 3D}(u,\Fb)=
\int_{\Phi(V)} \det(g_{jk})^{1/2}\biggl[
  \sum_{1\leq j,k\leq 3} g^{jk}
  \bigl(\partial_{y_j}-i\kappa H\widetilde{F}_j\bigr)\widetilde u \times
  \overline{\bigl(\partial_{y_k}-i\kappa H\widetilde{F}_k\bigr)\widetilde u}\\
  -\kappa^2 |\widetilde u|^2
  +\frac{\kappa^2}{2}|\widetilde u|^4\biggr]\,dy,
\end{multline}
where $\widetilde u = e^{i \kappa H \phi} u \circ  \Phi^{-1}$ (with $\phi$ the
gauge transformation necessary to pass to the ${\bf \widetilde F}$ given
in~\eqref{eq:tF}).

\subsection{Decomposition of the energy}
Consider a family $\{x_{0,l}\}_l$ of points on the boundary $\partial\Omega$
whose choice will be specified below.
For each point $x_{0,l}$, we may  introduce a coordinate transformation
$\Phi_l$ valid near the point $x_{0,l}$. Actually, after performing a
translation, we may reduce to the case corresponding to the coordinate
transformation $\Phi$ valid near  the point $x_0\in\partial\Omega$ as above.

We introduce a new partition of unity $\{\tchi_l\}_{l}$ covering
the set $\Omega_1:=\supp\chi_1$, whose support will be centered at
coordinate neighborhoods of a family of points $\{x_{0,l}\}$. Let
$\alpha=\alpha(\kappa)\ll 1$ be a parameter that will be
explicitly chosen below. Let, for $\delta>0$,
\[
O_{\delta} := \bigl\{(y_1,y_2,y_3)\bigm| 0<y_1,\
-\delta <y_2<\delta,\ -\delta < y_3 < \delta\bigr\}\,.
\]
We choose $\{\tchi_l\}_{l}$ as smooth non-negative functions such that
$\sum_l \tchi_l^2(x)\equiv 1$ in $\Omega_1$, $\tchi_l\equiv 1$
in the set $\Omega_1\cap\Phi_l^{-1}\bigl(O_{(1-\alpha)\delta}\bigr)$,
and $\supp\tchi_l \subset \Omega_1\cap \Phi_l^{-1}\bigl(O_{\delta}\bigr)$,
and such that there exists a constant $C>0$ so that
\[
\sum_l |\nabla \tchi_l(x)|^2 \leq C(\alpha\delta)^{-2}
\]
and
$$
\sum \tchi_l(x)^2 \leq C \qquad \text{(finite overlap)}
$$
for all $x\in\Omega_1$.

Using the IMS decomposition formula and the inequality
$\int_{\Omega} (\tchi_l(x)^2-\tchi_l(x)^4)|\chi_1\psi|^4\,dx \geq 0$, we get
the following lower bound of the surface energy,
\begin{equation}\label{eq:n}
\mathcal E_0^{\rm 3D}(\chi_1h\psi,\Ab) \geq
\sum_{l} \biggl\{ \mathcal E_0^{\rm 3D}(\tchi_lh\chi_1\psi,\Ab)
-\int_{\Omega_1}|\nabla \tchi_l|^2|\chi_1h\psi|^2\, dx\biggr\}.
\end{equation}
Using the bound on $\nabla\tchi_l$, we bound the error term by
\[
\sum_{l}\int_{\Omega_1}|\nabla \tchi_l|^2|\chi_1h\psi|^2\, dx
\leq C\alpha^{-2}\delta^{-1} \,.
\]

We approximate the magnetic potential $\Ab$ by the magnetic potential $\Fb$.
Part of the approximation relies on the construction of a suitable gauge
transformation. Let
\[
\phi_l(x)=\big(\Ab(x_{0,l})-\Fb(x_{0,l})\big)\cdot x\,,
\]
and
\[
u_l(x)=e^{i\phi_l(x)}\,\tchi_l(x)\chi_1(x)h(x)\psi(x)\,.
\]
Then it holds true that,
\begin{equation}\label{eq-gauge}
\mathcal E_0^{\rm 3D}(\tchi_l\chi_1h\psi,\Ab)
=\mathcal E_0^{\rm 3D}(u_l,\Ab-\nabla\phi_l)\,.
\end{equation}
We will prove the following estimate (for large values of $\kappa$),
\begin{equation}\label{eq-AtoF}
\mathcal E_0^{\rm 3D}(u_l,\Ab-\nabla\phi_l)\geq
(1-\delta )\mathcal E_0^{\rm 3D}(u_l,\Fb)
-C\delta \kappa^2\int_{\Omega}|u_l|^2\,dx\,.
\end{equation}
The proof of~\eqref{eq-AtoF} will be postponed to the end of the section.
We will bound the energy $\mathcal E_0(u_l,\Fb)$ from below by expressing it
in boundary coordinates. Let $\nu=\nu_l=\nu(x_{0,l})$ be the angle
between the applied magnetic field $\beta=(0,0,1)$ and the tangent plane of
$\partial\Omega$ at the point $x_{0,l}$. Selecting the gauge giving us the
magnetic potential $\widetilde\Fb$ in~\eqref{eq:tF}, we get with
\[
\psi_l=u_l\circ\Phi_l \times (\text{a gauge transformation}),
\]
that
\begin{multline*}
\mathcal E_0^{\rm 3D}(u_l,\Fb)=
\int_{O_{\delta}} \det(g_{jk})^{1/2}\biggl[
  \sum_{1\leq j,k\leq 3} g^{jk}
  \bigl(\partial_{y_j}-i\kappa H\widetilde{F}_j\bigr)\psi_l \times
  \overline{\bigl(\partial_{y_k}-i\kappa H\widetilde{F}_k\bigr)\psi_l}\\
  -\kappa^2 |\widetilde\psi_l|^2
  +\frac{\kappa^2}{2}|\widetilde\psi_l|^4 \biggr]\,dy
\end{multline*}

Inserting the estimates~\eqref{eq:metricapprox} and~\eqref{eq:jacobian}
we obtain (again it is assumed that $\delta$ is
sufficiently small)
\begin{align*}
\mathcal E_0^{\rm 3D}(u_l,\Fb)
&\geq
(1-C\delta)\int_{O_{\delta}}
\left(|(\nabla_y-i\kappa H\widetilde{\mathbf{F}}(y))\psi_l|^2
             - \kappa^2|\psi_l|^2+\frac{\kappa^2}{2}|\psi_l|^4\right)\,dy \\
             &\quad - C \delta \kappa^2 \int_{O_{\delta}} |\psi_l|^2\,dy\,.
\end{align*}
Using the pointwise inequality (with $\varepsilon>0$ arbitrary)
\[
|(\nabla_y-i\kappa H\widetilde{\mathbf{F}})\psi_l|^2
\geq (1-\varepsilon)|(\nabla_y-i\kappa H\Eb_{\nu_l})\psi_l|^2
+(1-\varepsilon^{-1})(\kappa H)^2|(\widetilde{\mathbf{F}}-\Eb_{\nu_l})\psi_l|^2
\]
with $\Eb_\nu$ from~\eqref{eq-3D-Eb}, we obtain
\begin{multline*}
\mathcal E_0^{\rm 3D}(u_l,\Fb)
\geq
(1-C\varepsilon-C\delta)\mathcal E_0^{\rm 3D}(\psi_l,\Eb_{\nu_l})
-C\varepsilon^{-1}(\kappa H)^2
\int_{O_{\delta}}|(\widetilde{\mathbf{F}}-\Eb_{\nu_l})\psi_l|^2 \,dy \\
-C(\varepsilon\kappa^2  + \delta \kappa^2 )\int_{O_{\delta}}|\psi_l|^2\,dy\,.
\end{multline*}
We estimate the integral
$\displaystyle\int_{O_{\delta}}|(\widetilde{\mathbf{F}}
-\Eb_{\nu_l})\psi_l|^2 \,dy$ using~\eqref{eq:tF}. In this way, we find,
\[
\varepsilon^{-1}(\kappa H)^2
\int_{O_{\delta}}|(\widetilde{\mathbf{F}}-\Eb_{\nu_l})\psi_l|^2 \,dy
\leq C\varepsilon^{-1} (\kappa H)^2 \delta^4\int_{O_{\delta}}|\psi_l|^2\,dy.
\]
We conclude that (with the choice $\varepsilon=\kappa \delta^2$)
\[
\mathcal E_0^{\rm 3D}(u_l,\Fb)
\geq
(1-C\delta^2 \kappa-C\delta)\mathcal E_0^{\rm 3D}(\psi_l,\Eb_{\nu_l})
-(\delta + \delta^2 \kappa)\kappa^2 \int_{O_{\delta}}|\psi_l|^2\,dy\,.
\]
After a scaling, $y=(\kappa H)^{-1/2}z$, we obtain
\begin{equation}\label{eq-correction}
\mathcal E_0^{\rm 3D}(\psi_l,\Eb_{\nu_l})
=\frac{1}{\sqrt{\kappa H}}\int_{O_{\sqrt{\kappa H}\delta}}
  \left( |(\nabla_z-i\Eb_{\nu_l}(z))\widetilde\psi_l|^2
             - \frac{\kappa}{H}|\widetilde\psi_l|^2
             +\frac{\kappa}{2H}|\widetilde\psi_l|^4\right)\,dz\,,
\end{equation}
where $\widetilde\psi_l(z)=\psi_l\big((\kappa H)^{-1/2}z\big)$.

We impose the condition $\sqrt{\kappa H}\delta\gg 1$. Next, by
defining $\widetilde \lambda = [\kappa/H-1]_{+}$, we get $0<
\kappa/H\leq 1+\widetilde\lambda$.

Let $\bb=\min\big(\kappa/H,1)$. It is easy to see that,
\[-\frac{1}{\sqrt{\kappa H}}\int_{O_{\sqrt{\kappa H}\delta}}
\frac{\kappa}H|\widetilde\psi_l|^2\,dz\geq
-\frac1{\sqrt{\kappa H}} \int_{O_{\sqrt{\kappa
H}\delta}}\bb|\widetilde\psi_l|^2\,dz-C\widetilde \lambda(\kappa
H)\int_{O_{\delta}}|\psi_l|^2\,dz\,.\]
We insert this estimate
into~\eqref{eq-correction}, then we apply\,\footnote{We need to
apply Theorem~\ref{thm-thdL} when $\bb\in[\Theta_0,1]$. If
$\bb\in(0,\Theta_0)$, we use Remark~\ref{rem-b<theta0}.}
Theorem~\ref{thm-thdL} (with $\ell=\sqrt{\kappa H}\delta$). In
this way, we conclude that,
\begin{align}
\label{eq:nn}
\mathcal E_0^{\rm 3D}(u_l,\Fb)
\geq
(1-C\delta^2 \kappa-C\delta)\sqrt{\kappa H}\,E(\bb,\nu(x_{0,l}))\,(4\delta^2)
-C(\delta+ \delta^2 \kappa
+\widetilde \lambda )\kappa^2\int_{O_{\delta}}|\psi_l|^2\,dy\,.
\end{align}
provided that $\kappa$ is large enough.
We combine the estimates in~\eqref{eq:n}-\eqref{eq:nn}
to get,
\begin{equation*}
\begin{aligned}
\mathcal E_0^{\rm 3D}(\chi_1h\psi,\Ab) \geq&
\sum_{l} \Bigl\{ \mathcal E_0^{\rm 3D}(\tchi_l\chi_1h \psi,\Ab)\Bigr\}
- C\alpha^{-2}\delta^{-1}\\
\geq& \sum_{l}
\Bigl\{(1-C\delta^2 \kappa-C\delta)\sqrt{\kappa H}E(\bb,\nu(x_{0,l}))(4\delta^2)
\\&\qquad\qquad-C(\delta + \delta^2 \kappa
+ \widetilde \lambda) \kappa^2 \int_{O_{\delta}}|\psi_l|^2\, dy\Bigr\}\\
&\qquad - C\alpha^{-2}\delta^{-1}.\\
\end{aligned}
\end{equation*}
We estimate as before using the finite overlap of the supports of the partition
of unity,
\begin{align}
\sum_{\ell} \int_{O_{\delta}}|\psi_l|^2\, dy
\leq C \int_{\Omega_1} |\psi|^2 \,dx \leq C \delta.
\end{align}

Next, we note that we have a Riemann sum,
\[
\biggl|\sum_{l}
\Bigl\{ E(\bb,\nu(x_{0,l}))(4\delta^2)\Bigr\}
-\int_{\overline D\cap \partial\Omega} E(\bb,\nu(x))\,d\sigma(x)\biggr|
\leq C\delta\,.
\]
Hence, we find that
\begin{equation}\label{eq-finalF}
\begin{aligned}
\mathcal E_0^{\rm 3D}(\chi_1h\psi,\Ab) \geq&
(1-C\delta^2 \kappa-C\delta)\sqrt{\kappa H}
\int_{\overline D\cap\partial\Omega} E(\bb,\nu(x))\,d\sigma(x)\\
&-C \delta\sqrt{\kappa H}
- C( \delta^2\kappa + \delta+ \widetilde \lambda) \delta \kappa^2
-C\alpha^{-2}\delta^{-1}.
\end{aligned}
\end{equation}
We choose the parameters
\begin{align}\label{eq:deltaalpha}
\delta = \kappa^{-3/4},\qquad \alpha = \kappa^{-1/16},
\end{align}
so that the remainder terms in~\eqref{eq-finalF} are all estimated as
\[
o\Big(\max(\kappa,[\kappa-H]_+^2)\Big)\,,
\]
as $\kappa\to\infty$.

\begin{proof}[Proof of~\eqref{eq-AtoF}]
Let $\widetilde\Ab=\Ab-\nabla\phi_l$. Using the pointwise inequality
\[
|z_1-z_2|^2\geq (1-\varepsilon)|z_1|^2+(1-\varepsilon^{-1})|z_2|^2\,,
\]
which is valid for all complex
numbers $z_1$, $z_2$ and real numbers  $\varepsilon\in(0,1)$, we obtain the
lower bound
\begin{equation}\label{eq-feq}
\begin{aligned}
\mathcal E^{\rm 3D}(u_l,\widetilde\Ab)
&\geq (1-\varepsilon)\mathcal E_0^{\text{3D}}(u_l,\Fb)
+\varepsilon^{-1}(\kappa H)^2 \int_\Omega |(\widetilde\Ab-\Fb)u_l|^2\, dx
- \varepsilon \int_{\Omega}\kappa^2|u_l|^2\,dx.
\end{aligned}
\end{equation}
Observing that
\[
(\widetilde \Ab-\Fb)(x)=(\Ab-\Fb)(x)-(\Ab-\Fb)(x_{0,l})\,,
\]
we get from Corollary~\ref{corol:L2}
\begin{align*}
(\kappa H)^2 \int_\Omega |(\widetilde\Ab-\Fb)u_l|^2\, dx
&\leq C \lambda^{1/3}\kappa^2\int|x-x_{0,l}|^2\,|u_l|^2\,dx\\
&\leq C \lambda^{1/3}\kappa^2\delta^2\int|u_l|^2\,dx\,.\end{align*}
Inserting this into~\eqref{eq-feq}, and making for simplicity the non-optimal
choice $\varepsilon = \delta$, we finish the proof of~\eqref{eq-AtoF} by
observing that $\lambda \leq 1$ for large values of $\kappa$.
\end{proof}

\subsection{The bulk term}
We return to the energy decomposition in~\eqref{eq-splitting}. The
choice we made for the parameter $\delta$ allows us to estimate
the upper bound in~\eqref{eq-B+S} as ${\mathcal O}(\kappa^{3/4}) = o(\kappa)$. The surface term
in~\eqref{eq-splitting} is estimated using~\eqref{eq-finalF}.
It only remains to estimates the bulk term $\mathcal
E_0(\chi_2h\psi,\Ab)$ appearing in~\eqref{eq-splitting}. The
estimate of this term was the objective of the first part of this
paper \cite{FK3D}. Actually, Theorem~6.1 of \cite{FK3D} tells us
that,
\[
\mathcal E_0(\chi_2h\psi,\Ab)\geq E_2|D|\,[\kappa-H]_+^2
+o\Big(\max(\kappa,[\kappa-H]_+^2)\Big)\,.
\]
Inserting this estimate into~\eqref{eq-splitting}, then using the lower
bound in~\eqref{eq-finalF} and the choice of $\delta$ finishes the proof of
Theorem~\ref{thm-lb}.

\section{Upper bound}\label{sec:ub}

The aim of this section is to give an asymptotic upper bound of the ground state energy in \eqref{eq-3D-gs}.
This will be done through the computation of the energy of relevant  test configurations, whose construction hints at the expected behavior of the actual minimizers of the energy in \eqref{eq-3D-GLf}.

The main theorem of this section is stated below.

\begin{thm}\label{thm-ub}
For all $\tilde a(\kappa) = o(1)$ as $\kappa \rightarrow \infty$ and $C>0$,
there exists $0< \err(\kappa) = o(1)$ such that if
$$
1 - \tilde a(\kappa) \leq \frac{H}{\kappa} \leq C.
$$
Then, as $\kappa\to\infty$, the ground state energy in \eqref{eq-3D-gs} satisfies,
$$
\E0
(\kappa,H)
\leq \sqrt{\kappa H} \int_{\partial\Omega}E(\bb,\nu(x))\,d\sigma(x)+E_2|\Omega|\,[\kappa-H]_+^2+\err\big(\max(\kappa,[\kappa-H]_+^2)\big)\,.
$$
Here $\bb=\min\big(\kappa/H,1\big)$, and $d\sigma(x)$ is the surface measure on the boundary of $\Omega$.
\end{thm}

\medskip

\begin{proof}[Proof of Theorem~\ref{thm-ub}]~\\
{\bf Boundary trial configuration}\\
Let $\delta>0$ be small but fixed. We will choose another parameter $\eta>0$ which will be specified as a negative power of $\kappa$ below.

Choose a finite collection of points $\{x_j\} \subset \partial \Omega$ such that
\begin{align*}
\forall j \neq k: \quad \delta/2 \leq \dist(x_j,x_k)\qquad
\text{ and }
\qquad
\forall j:\quad \min_{k\neq j} \dist(x_j,x_k) \leq 2\delta.
\end{align*}
Define $U_j$ as
$$
U_j=\{x\in\partial\Omega~:~\forall~k\not=j\,,~{\rm
dist}(x,x_j)<{\rm dist}(x,x_k)\}\,.
$$
Clearly the $U_j$'s are disjoint and $\partial \Omega = \bigcup_j \overline{U}_j$.

Next, we construct a
family of sets that covers a tubular neighborhood of
$\partial\Omega$. That will be done by using  the boundary
coordinates $(y_1,y_2,y_3)$ introduced in Sec.~\ref{sec:bndcod}
($y_1=0$ defines the corresponding part in $\partial\Omega$). Let
$\Phi_j$ be the coordinate transformation that straightens a
neighborhood $V_j$ of the point $x_j$ such that $\Phi_j(x_j)=0$.
We may assume that $U_j \subset V_j$ for all $j$.
Let
$$O_j=\{x=\Phi_j^{-1}(y_1,y_2,y_3)~:~\Phi_j^{-1}(0,y_2,y_3)\in U_j~{\rm and}~0<y_1<\eta\}\,.$$

We now choose a parameter $\delta'$:
$$
\frac{1}{\sqrt{\kappa H}} \ll \delta' \ll \delta
$$
($\delta'$ will be chosen below as a negative power of $\kappa$).
Define $\widetilde O_j^{\rm 2D}=\Phi_j(O_j)\cap \{ y \in {\mathbb R}^3 \,:\, y_1 = 0\}$.
We may cover $\widetilde O_j^{\rm 2D}$ by a square
lattice $\{K_{j,i}\}$ where each $K_{j,i}$ is centered at point
$y_{j,i}$ and has side-length $2\delta'$. Let
$$\mathcal J_j=\{i~:~K_{j,i}\subset \widetilde O_j^{\rm 2D}\},\quad
N_j={\rm Card}\,\mathcal J_j\,,
$$
and $N=\sum_jN_j$. Clearly, the number $N$ satisfies,
\begin{equation}\label{eq-N=card-J}
N\times(2\delta')^2\to |\partial\Omega|\quad{\rm as
~}\delta'\to0\,.
\end{equation}
We combine the coordinate transformation $\Phi_j$ by a translation
so that the new coordinates of the point $y_{j,i}$ are $0$. Thus, we
let $\Phi_{j,i}$ be the resulting coordinate transformation valid in
$\Phi_j^{-1}({\mathbb R}^+\times K_{j,i}) \cap V_j$ such that $\Phi_{j,i}(y_{j,i})=0$.

We consider only indices $i\in\mathcal J_j$ for some $j$. Let
$x_{j,i}=\Phi_{j,i}^{-1}(0)$. At each point $x_{j,i}$,
 the magnetic field $\beta=(0,0,1)$ forms an angle
$\nu_{j,i}=\nu(x_{j,i})\in[0,\pi/2]$ with the tangent plane to
$\partial\Omega$.  As explained earlier in Sec.~\ref{sec:bndcod},
there exists a real valued function $\phi_{j,i}$ such that, if
$\Fb$ is the vector field defined in cartesian coordinates by
$\Fb(x_1,x_2,x_3)=(-x_2/2,x_1/2,0)$, and $\widetilde \Fb$ is the
vector field defined in $y$-coordinates by the relation in
\eqref{eq-gauge-F}, then,
$$
\widetilde \Fb+\nabla\phi_{j,i}=\Eb_{\nu_{j,i}}+\Rb_{j,i}\,,$$
where $\Eb_{\nu_{j,i}}$ is the magnetic potential from
\eqref{eq-3D-Eb}, and $\Rb_{j,i}$ is a vector field given in
$y$-coordinates by,
$$\Rb_{j,i}=\left(
\begin{array}{c}
0\\
\mathcal O(|y_1|^2+|y_2|^2+|y_3|^2)\\
\mathcal O(|y_1|^2+|y_2|^2+|y_3|^2)
\end{array}\right)\,.
$$
Let $\bb=\min(\kappa/H,1)$, $\ell=\delta'\sqrt{\kappa H}$ and
$u_{j,i}$ a minimizer of the functional $\mathcal
G_{\bb,\nu_{j,i},\ell}$ introduced in \eqref{eq-Rgl}. For
$x=\Phi_{j,i}^{-1}(y_1,y_2,y_3)\in O_j$, we put,
$$
\psi_{j,i}(x)=\begin{cases}e^{-i\kappa H \phi_{j,i}} u_{j,i}\big(y\sqrt{\kappa H}\,\big), & x=\Phi_{j,i}^{-1}(y_1,y_2,y_3)\in O_j\\
0, & \text{ else }
\end{cases}\,.
$$
Notice that by construction the supports of the $\psi_{j,i}$ do not overlap.

We define a test function $\psi_{\rm bnd}\in H^1(\Omega;\C)$ as follows,
\begin{equation}\label{eq-testf-bnd}
\forall~x\in\Omega\,,\quad \psi_{\eta,\delta}^{\rm bnd}(x)=h\left(\frac{{\rm dist}(x,\partial\Omega)}{\eta}\right)\psi(x)\,,
\end{equation}
where  $h$ is a cut-off function satisfying,
$${\rm supp}\,h\subset [-1,1]\,,\quad 0\leq h\leq 1{\rm ~in~}\R\,,\quad h(x)=1 {\rm ~in}~[-1/2,1/2]\,,$$
and
$$
\psi = \sum_{i,j} \psi_{j,i}.
$$
\medskip

\noindent{\it Energy of the test-configuration}\\
We will compute the energy of the configuration $(\psi_{\eta,\delta}^{\rm bnd},\Fb)$. Notice that the construction of $\psi_{\eta,\delta}^{\rm bnd}$ and the change of variable formulas in Sec.~\ref{sec:bndcod} together imply the existence of a constant $C>0$ (independent of $\delta$) such that, if $\kappa$ is sufficiently large and $\beta$ is an arbitrary real number in $(0,1)$, then we have the upper bound,
\begin{align}\label{eq-ub-c0}
\mathcal E^{\rm 3D}(\psi_{\eta,\delta}^{\rm bnd},\Fb)&\leq \left(1+C(\delta+\eta+\beta)\right)
\sum_{j,i}
\int_{V_{\delta'}}|(\nabla-i\kappa H \Eb_{\nu_{j,i}})\,h(y_1/\eta)\,u_{j,i}(y\sqrt{\kappa H})|^2\,dy\nonumber\\
&\quad+C\kappa^2H^2\beta^{-1}(\delta'^4+\eta^4)\sum_{j,i}\int_\Omega|\psi_{j,i}(x)|^2\,dx
+\sum_{j,i}\int_{\Omega}\left(-\kappa^2|\psi_{j,i}|^2+\frac{\kappa^2}2|\psi_{j,i}|^4\right)\,dx\,,
\end{align}
where
$V_{\delta'}=(0,\infty)\times(-\delta',\delta')\times(-\delta',\delta')$.

By Theorem~\ref{thm-min-Rgl}
$$\int\left(|u_{j,i}(z)|^2+|u_{j,i}(z)|^4\right)\,dz\leq C\ell^2\,,$$
where $C$ is a universal constant and
 $\ell=\delta'\sqrt{\kappa H}$. We use  the change of variable formulae in Sec.~\ref{sec:bndcod}
 to express $$-\int|\psi_j|^2\,dx+\frac12\int|\psi_j|^4\,dx$$
in boundary coordinates, then we apply the change of variable
$z=y\sqrt{\kappa H}$ and use the decay of $u_{j,i}$. In this way,
we get a constant $C$ such that, for all $j$ and $i\in\mathcal
J_j$, we have,
$$-\kappa^2\int|\psi_{j,i}|^2\,dx
+\frac{\kappa^2}2\int|\psi_{j,i}|^4\,dx \leq \sqrt{\kappa
H}\int_{U_\ell}\left(-\frac{\kappa}H|u_{j,i}(z)|^2
+\frac{\kappa}{2H}|u_{j,i}|^4\right)\,dx+\frac{C\kappa^2}{\sqrt{\kappa
H}}(\delta+\eta)(\delta')^2\,.
$$
Here $U_\ell=(0,\infty)\times(-\ell,\ell)\times(-\ell,\ell)$.
Also,
$$
C\kappa^2H^2\beta^{-1}(\delta'^4+\eta^4)\int_\Omega|\psi_{j,i}(x)|^2\,dx
\leq
C \kappa^2 \frac{H^2}{\sqrt{\kappa H}} \beta^{-1} (\delta'^4+\eta^4) (\delta')^2.
$$

Using again Theorem~\ref{thm-min-Rgl}
\begin{align*}
&\int_{V_{\delta'}}|(\nabla-i\kappa H
\Eb_{\nu_{j,i}})\,h(y_1/\eta)\,u_{j,i}(y\sqrt{\kappa
H})|^2\,dy\\
&\leq
(1+\beta) \int_{V_{\delta'}}|(\nabla-i\kappa H
\Eb_{\nu_{j,i}})u_{j,i}(y\sqrt{\kappa
H})|^2\,dy
+ \beta^{-1} \eta^{-2} \int_{V_{\delta'}}[h'(y_1/\eta)]^2\,|u_{j,i}(y\sqrt{\kappa
H})|^2\,dy\\
&\leq (1+\beta) \sqrt{\kappa H}\int_{U_\ell}|(\nabla-i\Eb_{\nu_{j,i}})u_{j,i}|^2\,dz
+ \beta^{-1} \eta^{-3} (\kappa H)^{-2} \int_{U_{\ell}} z_1 |u_{j,i}|^2 \,dz \\
&\leq (1+\beta) \sqrt{\kappa H}\int_{U_\ell}|(\nabla-i\Eb_{\nu_{j,i}})u_{j,i}|^2\,dz
+
\beta^{-1} \eta^{-3} (\kappa H)^{-1} (\delta')^2.
\end{align*}
Substituting the above upper bounds into \eqref{eq-ub-c0}, we get,
\begin{align}\label{eq-ub-c0-1}
&\mathcal E^{\rm 3D}(\psi_{\eta,\delta}^{\rm bnd},\Fb)\nonumber \\
&\leq \sum_{j,i}\Bigg\{
[1+C(\delta + \eta+\beta)]\sqrt{\kappa H}
\int_{U_\ell}\left(|(\nabla-i\Eb_{\nu_{j,i}})u_{j,i}|^2-
\frac{\kappa}H|u_{j,i}|^2+\frac{\kappa}{2H}|u_{j,i}|^4\right)\,dz\nonumber \\
&
+C \kappa^2 \frac{H^2}{\sqrt{\kappa H}} \beta^{-1} (\delta'^4+\eta^4) (\delta')^2
+ C \kappa^2 \frac{1}{\sqrt{\kappa H}} (\delta + \eta + \beta) (\delta')^2
+ C \beta^{-1} \eta^{-3} (\kappa H)^{-1} (\delta')^2 \Bigg\}.
\end{align} Recall that
$\bb=\min(\kappa/H,1)$. By our choice of $u_{j,i}$, we get that,
$$\int_{U_\ell}\left(|(\nabla-i\Eb_{\nu_{j,i}})u_{j,i}|^2-\frac{\kappa}H|u_{j,i}|^2+\frac{\kappa}{2H}|u_{j,i}|^4\right)\,dz\leq
(2\delta')^2 \left[E(\bb,\nu_{j,i})+o(1)\right]\,.
$$
Also, by \eqref{eq-N=card-J}, we have
$$
\sum_{i,j} (\delta')^2 \leq C.
$$
Inserting this and the bounds on $H/\kappa$ into \eqref{eq-ub-c0-1}, we get
\begin{align}
\mathcal E^{\rm 3D}(\psi_{\eta,\delta}^{\rm bnd},\Fb)
&\leq
[1+C(\delta + \eta+\beta)]\sqrt{\kappa H} \sum_{j,i} (2\delta')^2 \left[E(\bb,\nu_{j,i})+o(1)\right] \nonumber \\
&\quad+ C \big[ \kappa^3 \beta^{-1} (\delta'^4+\eta^4)
+  \kappa (\delta + \eta + \beta)
+ \beta^{-1} \eta^{-3} \kappa ^{-2} \big] .
\end{align}
We can now, for example, choose
$$
\eta = \delta' = \kappa^{-7/12},\qquad \beta = \kappa^{-1/6}.
$$
The sum is a
Riemann sum, hence as
$\delta'\to0$,
$$
\sum_{j,i} (2\delta')^2 E(\bb,\nu_{j,i})
=
\int_{\partial\Omega}E(\bb,\nu(x))\,d\sigma(x)+o(1)\,.
$$
Therefore, we get the following upper bound,
\begin{equation}\label{eq-ub-bnd}
\mathcal E^{\rm 3D}(\psi_{\eta,\delta}^{\rm bnd},\Fb)\leq
(1+C \delta) \sqrt{\kappa
H}\int_{\partial\Omega}E(\bb,\nu(x))\,d\sigma(x)+o(\kappa) + C\delta \kappa\,.
\end{equation}
Since $\delta >0$ was arbitrary this is consistent with the boundary part of the upper bound in Theorem~\ref{thm-ub}.

\noindent
{\bf Bulk trial configuration}\\
We keep the choice of the parameters $\delta$ and $\eta$ introduced in the preceding section. Let $\psi_{\eta,R}^{\rm blk}$ be the test function defined in \cite[Eq. (6.15)]{FK3D}. The function $\psi_{\eta,R}^{\rm blk}$ vanishes in
$$\{x\in\Omega~:~{\rm dist}(x,\partial\Omega)\leq \eta\}\,.$$
Thus, $\psi_{\eta,R}^{\rm blk}$ and $\psi_{\eta,\delta}^{\rm bnd}$ have disjoint support. Consequently, by defining,
$$f(x)=\psi_{\eta,R}^{\rm blk}(x) +\psi_{\eta,\delta}^{\rm bnd}(x)\,,\quad x\in\Omega\,,$$
we get,
\begin{equation}\label{eq-ub-blk+bnd}
\mathcal E^{\rm 3D}(f,\Fb)=\mathcal E^{\rm 3D}(\psi_{\eta,R}^{\rm blk},\Fb)+\mathcal E^{\rm 3D}(\psi_{\eta,\delta}^{\rm bnd},\Fb)\,.\end{equation}
Theorem~6.5 in \cite{FK3D} tells us that
$$ \mathcal E^{\rm 3D}(\psi_{\eta,R}^{\rm blk},\Fb)\leq
E_2|\Omega|\,[\kappa-H]_+^2+o\bigg(\max\big(\kappa,[\kappa-H]_+^2\big)\bigg)\,.$$
Inserting this estimate and  that in \eqref{eq-ub-bnd} into
\eqref{eq-ub-blk+bnd}, we get,
$$\mathcal E^{\rm 3D}(f,\Fb)\leq
\sqrt{\kappa
H}\int_{\partial\Omega}E(\bb,\nu(x))\,d\sigma(x)+E_2|\Omega|\,[\kappa-H]_+^2
+o\bigg(\max\big(\kappa,[\kappa-H]_+^2\big)\bigg)\,.$$ Recalling
the ground state energy in \eqref{eq-3D-gs}, we deduce the
following upper bound,
$$\E0
(\kappa,H)\leq \sqrt{\kappa
H}\int_{\partial\Omega}E(\bb,\nu(x))\,d\sigma(x)+E_2|\Omega|\,[\kappa-H]_+^2
+o\bigg(\max\big(\kappa,[\kappa-H]_+^2\big)\bigg)\,.$$ This
finishes the proof of Theorem~\ref{thm-ub}.
\end{proof}

\section{Proof of main theorems}

\subsection{Proof of Theorem~\ref{thm-main}}
We  combine the lower bound of Theorem~\ref{thm-lb} (with $D=\Omega$ and $h=1$)
and the upper bound of Theorem~\ref{thm-ub}.

\subsection{Proof of Theorem~\ref{thm-op}}
Let $D\subset \Omega$ be  open and smooth. Let $\mu:\R_+\to\R$ be a function
satisfying $\lim_{\kappa\to\infty}\mu(\kappa)=0$. Suppose that the magnetic
field $H$ satisfies $H\geq\kappa-\mu(\kappa)\kappa$.

Using Theorem~3.3 in \cite{FK3D}, we know that,
\begin{equation}\label{eq-FK3D}
\|\psi\|_{L^\infty(\omega_\kappa)}=o(1)\,,\quad\text{as }\kappa\to\infty\,,
\end{equation}
where
\[
\omega_\kappa=\{x\in\Omega~:~\dist(x,\partial\Omega)\geq g_1(\kappa)/\kappa\}\,,
\]
and $g_1(\kappa)$ is any function satisfying
$\lim_{\kappa\to\infty}g_1(\kappa)=\infty$.

\begin{lem}\label{lem-hc2-IMS}
If $(\psi,\Ab)$  is a solution of~\eqref{eq-3D-GLeq}, then,
\begin{align}
  \label{eq:10}
\mathcal E(\psi,\Ab;D)-\mathcal
E(f\psi,\Ab;D)-\int_D|\nabla f|^2|\psi|^2\,dx
&+\frac{\kappa^2}2\int_D(1-f^2)^2|\psi|^4\,dx\nonumber \\
&=- \Real\int_{\partial \Omega} |\psi|^2 \overline{ f} \;\nu
\cdot \nabla f\,d\sigma+o(\kappa)\,.
\end{align}
holds true as $\kappa\to\infty$.
Here $\nu$ is the unit inward normal vector of $\partial\Omega$ and
$f\in H^1(\Omega)$ is any function such that,
\[
\nabla f\in L^\infty(\Omega)\,,\quad \supp f\subset \overline D\,.
\]
\end{lem}
\begin{proof}
Using the estimate in~\eqref{eq-FK3D}, the result of the lemma follows
through an integration by parts applied to the term
$\mathcal E(f\psi,\Ab;D)$. See \cite[Lemma~6.1]{FK} for details.
\end{proof}

\begin{lem}\label{lem-hc2-curl}
If $(\psi,\Ab)$ is a minimizer of the functional in~\eqref{eq-3D-GLf}, then
as $\kappa\to\infty$,
\[
\kappa^2H^2\int_\Omega|\curl(\Ab-\Fb)|^2\,dx
=o\big(\max(\kappa,[\kappa-H]_+^2)\big)\,.
\]
\end{lem}
\begin{proof}
Recall the functional $\mathcal E_0$ in~\eqref{eq-GLe0}. Theorem~\ref{thm-lb}
tells us that,
\[
\mathcal E_0(\psi,\Ab)\geq
\int_{\partial\Omega}E(\bb,\nu(x))\,d\sigma(x)+E_2|\Omega|\,[\kappa-H]_+^2
+o\big(\max(\kappa,[\kappa-H]_+^2)\big)\,.
\]
Consequently, the result of the lemma follows by observing that
\begin{align*}
\mathcal E(\psi,\Ab)
=\mathcal E_0(\psi,\Ab)+\kappa^2H^2\int_{\R^3}|\curl(\Ab-\Fb)|^2\,dx
\end{align*}
and using the upper bound obtained in Theorem~\ref{thm-ub}.
\end{proof}

\begin{proof}[Proof of~\eqref{eq-3D-op'}]
Let $(\psi,\Ab)$ be a solution of~\eqref{eq-3D-GLeq}. Multiplying the first
equation in~\eqref{eq-3D-GLeq} by $\overline\psi$ then integrating over $D$ we
get,
\begin{equation}\label{eq-3D-op'-p}
\int_D|(\nabla-i\kappa H\Ab)\psi|^2\,dx
+\int_{\partial D}\overline\psi\,\nu\cdot(\nabla-i\kappa H)\psi\,d\sigma(x)=
\kappa^2\int_D(1-|\psi|^2)|\psi|^2\,dx\,.\end{equation}
Using~\eqref{eq-FK3D}, we can show that (see \cite[Lemma~6.1]{FK} for details),
\[
\int_{\partial D}\overline\psi\,\nu\cdot(\nabla-i\kappa H)\psi\,d\sigma(x)
=o(\kappa)
\]
as $\kappa\to\infty$. Consequently, we may rewrite~\eqref{eq-3D-op'-p} as
follows,
\begin{equation}\label{eq-3D-op'-p1}
\mathcal E_0(\psi,\Ab;D)=
-\frac{\kappa^2}2\int_D|\psi|^4\,dx+o(\kappa)\,,
\end{equation}
where $\mathcal E_0$ is the functional introduced in~\eqref{eq-GLe0}.

Let $h\in C_c^\infty(\R)$ be a cut-off function such that
\[
\supp h\subset[-1,1]\,,\quad h(x)=1\text{ in }[-1/2,1/2]\,,\quad 0
\leq h\leq 1\text{ in }\R\,.
\]
We define a function $f$ as follows,
\begin{equation}\label{eq-3D-op'-f}
\forall~x\in\Omega\,,\quad
f(x)=1-h\left(\frac{\dist(x,D^{\mathrm c})}{L^{-1}}\right)\,,
\end{equation}
where $D^{\mathrm c}=\Omega \setminus D$ is the complement in $\Omega$ of $D$,
and
\begin{equation}\label{eq(L)}
L=\max(\kappa,[\kappa-H]_+^2)\,.
\end{equation}
We will prove the estimate below,
\begin{equation}\label{eq-3D-op'-r}
-\int_D|\nabla
f|^2|\psi|^2\,dx+\frac{\kappa^2}2\int_D(1-f^2)^2|\psi|^4\,dx
+\Real\int_{\partial \Omega} |\psi|^2 \overline{ f} \;\nu \cdot
\nabla
f\,d\sigma=o\big(\max(\kappa,[\kappa-H]_+^2)\big)\,.\end{equation}
Details concerning the derivation of the estimate in
\eqref{eq-3D-op'-r} will be postponed to the end of this proof.

Using~\eqref{eq-3D-op'-r} and the decomposition formula of
Lemma~\ref{lem-hc2-IMS}, we get,
\begin{equation}\label{eq-3D-op'-p2}\mathcal E_0(\psi,\Ab;D)=
\mathcal E_0(f\psi,\Ab;D)+o\big(\max(\kappa,[\kappa-H]_+^2)\big)\,.
\end{equation}
We apply Theorem~\ref{thm-lb} to bound $\mathcal E_0(f\psi,\Ab;D)$ from below.
In this way, we infer from~\eqref{eq-3D-op'-p2},
\begin{equation}\label{eq-3D-op'-p3}
\mathcal E_0(\psi,\Ab;D)\geq
\sqrt{\kappa H}\int_{\overline{D}\cap\partial\Omega}E(\bb,\nu(x))\,d\sigma(x)
+E_2|D|\,[\kappa-H]_+^2
+o\big(\max(\kappa,[\kappa-H]_+^2)\big)\,.
\end{equation}
Inserting this lower bound into~\eqref{eq-3D-op'-p1} finishes the
proof of~\eqref{eq-3D-op'}. The only point left is the
justification of the estimate in~\eqref{eq-3D-op'-r}.

\paragraph{\it Proof of~\eqref{eq-3D-op'-r}:}
Let $g_1(\kappa)=L^{1/2}$. Recall the estimate in~\eqref{eq-FK3D}
valid in the set $\omega_\kappa$. In order to estimate the
integral $\displaystyle\int_D(1-f^2)^2|\psi|^4\,dx$, we use the
simple decomposition,
\[
\int_D(1-f^2)^2|\psi|^4\,dx=\int_{\omega_\kappa\cap D}(1-f^2)^2|\psi|^4\,dx
+\int_{\omega_\kappa^c\cap D}(1-f^2)^2|\psi|^4\,dx\,.
\]
We estimate the integral over $\omega_\kappa\cap D$ as follows:
\begin{align}
\int_{\omega_\kappa\cap D}(1-f^2)^2|\psi|^4\,dx
\leq o\Big( \int_{D}(1-f^2)^2\,dx\Big)
\leq o(L^{-1}).
\end{align}
It is left to estimate the integral over $\omega_\kappa^c\cap D$.
Here we use that the measure of the set
\[
\supp(1-f^2) \cap \omega_\kappa^c\cap D\subseteq \{x\in
\Omega~:~\dist(x,\partial\Omega)\leq
g_1(\kappa)/\kappa\,,~\dist(x,D^c)\leq 2L^{-1}\}\,\]
is
\[
\mathcal O\left(\frac{g_1(\kappa)}{\kappa}\times
2L^{-1}\right)\,.\]
As a consequence, we get by using the bound
$\|\psi\|_{L^\infty(\Omega)}\leq1$ together with our choice of
$g_1(\kappa)=L^{1/2}$ that,
\[
\int_{\omega_\kappa^c\cap D}(1-f^2)^2|\psi|^4\,dx=\mathcal
O\big(L^{-1/2}\kappa^{-1}\big)\,.
\]
The estimates obtained for the
integrals over $\omega_\kappa\cap D$ and $\omega_\kappa^c\cap D$
together yield that,
\begin{equation}\label{proof-7.6}
\kappa^2\int_D(1-f^2)^2|\psi|^4\,dx=o(\kappa)\,.\end{equation} In
a similar fashion, we may show that,
\[
\int_D|\nabla f|^2|\psi|^2\,dx=o(L)\,,\quad
\Real\int_{\partial \Omega} |\psi|^2 \overline{ f} \;\nu \cdot
\nabla f\,d\sigma=o(L)\,.
\]
The two aforementioned estimates
above, together with that in~\eqref{proof-7.6} and the definition
of $L=\max(\kappa,[\kappa-H]_+^2)$, yield the estimate in
\eqref{eq-3D-op'-r}.
\end{proof}

\begin{proof}[Proof of~\eqref{eq-3D-op-Thm}]
Suppose that $(\psi,\Ab)$ is a minimizer of the functional in~\eqref{eq-3D-GLf}.
Let the function $f$ be as in~\eqref{eq-3D-op'-f}. Since $(\psi,\Ab)$ is a
solution of~\eqref{eq-3D-GLeq}, then the formulas  in~\eqref{eq-3D-op'-p1}
and~\eqref{eq-3D-op'-p2} still hold true.

We define a function $g\in H^1(\Omega)$ as follows,
\[
\forall~x\in\Omega\,,\quad g(x)=\sqrt{1-f^2(x)}\,,
\]
so that $f^2(x)+g^2(x)=1$ in $\Omega$. Consequently, we have the  decomposition
formula,
\begin{equation}\label{eq-3D-op'-IMS}
\mathcal E(\psi,\Ab)\geq \mathcal E(f\psi,\Ab)+\mathcal E(g\psi,\Ab)
-\int_\Omega\left(|\nabla f|^2+|\nabla g|^2\right)|\psi|^2\,dx\,.
\end{equation}
Is results from Lemma~\ref{lem:almog} that as $\kappa\to\infty$,
\begin{equation}\label{eq-3D-op'-g}
\int_\Omega(|\nabla f|^2+|\nabla g|^2)|\psi|^2\,dx
=o\big(\max(\kappa,[\kappa-H]_+^2)\big)\,.
\end{equation}
Also, using Theorem~\ref{thm-lb}, we obtain the lower bound,
\[
\mathcal E(g\psi,\Ab)\geq \sqrt{\kappa H}
\int_{\overline {D^c}\cap\partial\Omega}E(\bb,\nu(x))\,d\sigma(x)
+E_2|D^c|\,[\kappa-H]_+^2+o\big(\max(\kappa,[\kappa-H]_+^2)\big)\,.
\]
We insert this estimate together with that in~\eqref{eq-3D-op'-g}
into~\eqref{eq-3D-op'-IMS}. Also, we use the upper bound for
$\mathcal E(\psi,\Ab)$ obtained in Theorem~\ref{thm-ub}. In this way we infer
from~\eqref{eq-3D-op'-IMS},
\begin{equation}\label{eq-3D-op'-IMS'}
\mathcal E(f\psi,\Ab)\leq \sqrt{\kappa H}\int_{\overline{ D}\cap\partial\Omega}E(\bb,\nu(x))\,d\sigma(x)
+E_2|D|\,[\kappa-H]_+^2+o\big(\max(\kappa,[\kappa-H]_+^2)\big)\,.
\end{equation}
Using the estimate of Lemma~\ref{lem-hc2-curl}, we get further,
\[
\mathcal E_0(f\psi,\Ab)\leq \sqrt{\kappa H}
\int_{\overline{ D}\cap\partial\Omega}E(\bb,\nu(x))\,d\sigma(x)
+E_2|D|\,[\kappa-H]_+^2+o\big(\max(\kappa,[\kappa-H]_+^2)\big)\,.
\]
We insert this estimate into~\eqref{eq-3D-op'-p2}. This gives us an upper
bound of $\mathcal E_0(\psi,\Ab;D)$, which when inserted
into~\eqref{eq-3D-op'-p1} yields the upper bound,
\[
-\frac{\kappa^2}2\int_D|\psi|^4\,dx\leq \sqrt{\kappa H}\int_{\overline{ D}\cap\partial\Omega}E(\bb,\nu(x))\,d\sigma(x)
+E_2|D|\,[\kappa-H]_+^2+o\big(\max(\kappa,[\kappa-H]_+^2)\big)\,.
\]
Combining this estimate with~\eqref{eq-3D-op'} finishes the proof
of~\eqref{eq-3D-op-Thm}.
\end{proof}

\section*{Acknowledgements}
The authors were partially supported by the Lundbeck foundation and the European
Research Council under the European Community's Seventh Framework
Programme (FP7/2007-2013)/ERC grant agreement n$^{\text{o}}$ 202859.  AK was
partially  supported by a grant from  the Lebanese University.

\end{document}